\documentclass[lettersize,journal]{IEEEtran}
\usepackage{amsmath,amsfonts}
\usepackage{algorithmic}
\usepackage{algorithm}
\usepackage{array}
\UseRawInputEncoding
\usepackage{textcomp}
\usepackage{stfloats}
\usepackage{url}
\usepackage{verbatim}
\usepackage{graphicx}
\usepackage{amssymb} % \triangleq
\usepackage{amsfonts}
\usepackage{amsmath} % align
\usepackage{physics} % 长竖线 
\usepackage{algorithmic}  % 算法框架
\usepackage{algorithm} 
\usepackage{graphicx}
\usepackage{color} % 字体颜色
\usepackage{bm}
\usepackage{epstopdf} 
\usepackage{subfigure}
\usepackage{booktabs}  % table
\usepackage[dvipsnames]{xcolor} % 更多字体颜色
\usepackage{cite} % 多项引用
\usepackage{balance}
\usepackage{setspace}  % 改变行距
\usepackage{amsthm}
\usepackage{amsmath} 
\usepackage{hyperref}
\usepackage{balance}
\usepackage{float}

\newtheoremstyle{note}
{3pt}%
{3pt}%
{}%
{\parindent}%
{\rmfamily \bfseries}%
{:}%
{5pt}%
{}%

\newtheorem{ppn}{Proposition}
 
\newtheorem{remark}{Remark}

\newtheorem*{pf}{Proof}  

\newcommand{\algref}[1]{\textbf{Algorithm \ref{#1}}}
\newcommand{\tabref}[1]{\textbf{Table \ref{#1}}}

\newcommand{\figref}[1]{Fig. \ref{#1}}
\newcommand{\secref}[1]{Section \ref{#1}}
\DeclareMathAlphabet{\mathsfit}{\encodingdefault}{\sfdefault}{m}{sl}
\SetMathAlphabet{\mathsfit}{bold}{\encodingdefault}{\sfdefault}{bx}{n}
\newcommand{\tens}[1]{\bm{\mathsfit{#1}}}
\def\tA{{\tens{A}}}
\def\tB{{\tens{B}}}
\def\tb{{\tens{b}}}
\def\tC{{\tens{C}}}

\def\tH{{\tens{H}}}

\def\tR{{\tens{R}}}

\def\tT{{\tens{T}}}
\def\tU{{\tens{U}}}

\def\tX{{\tens{X}}}

\begin{document}
	
	\title{Deep Learning-Based Multi-Satellite Massive MIMO Transmission: Centralized or Decentralized?}
	
	\author{Wenjing Cao, \textit{Graduate Student Member}, \textit{IEEE}, Yafei Wang, \textit{Graduate Student Member}, \textit{IEEE}, \\Jinshuo Zhang, \textit{Graduate Student Member}, \textit{IEEE},
      Xiaofan Xu, Wenjin Wang, \textit{Member}, \textit{IEEE},\\
   Symeon Chatzinotas, \textit{Fellow}, \textit{IEEE}, Bj{\"o}rn Ottersten, \textit{Fellow}, \textit{IEEE}% <-this % stops a space
		\thanks{Manuscript received xxx; revised xxx.  \textit{}}
		\thanks{Wenjing Cao, Yafei Wang, Jinshuo Zhang, and Wenjin Wang are with the National Mobile Communications Research Laboratory, Southeast University, Nanjing 210096, China, and also with Purple Mountain Laboratories, Nanjing 211100, China (e-mail: caowj@seu.edu.cn; wangyf@seu.edu.cn; zhangjs@seu.edu.cn; wangwj@seu.edu.cn).}
		\thanks{Xiaofan Xu is with the Shanghai Satellite Network Research Institute Co., Ltd., Shanghai 200120, China; Shanghai Key Laboratory of Satellite Network, Shanghai 200120, China; State Key Laboratory of Satellite Network, Shanghai 200120, China (e-mail: xiaofanxu@sina.com).}\thanks{Symeon Chatzinotas and Bj{\"o}rn Ottersten are with the Interdisciplinary Centre for Security, Reliability and Trust (SnT), University of Luxembourg, Luxembourg (e-mail: symeon.chatzinotas@uni.lu;  bjorn.ottersten@uni.lu).}}
	
	\markboth{}%
	{Shell \MakeLowercase{\textit{et al.}}: A Sample Article Using IEEEtran.cls for IEEE Journals}
	
	% 页尾
\maketitle
\begin{abstract}
This paper investigates new efficient transmission architectures for multi-satellite massive multiple-input multiple-output (MIMO). We study the weighted sum-rate maximization problem in a multi-satellite system where multiple satellites transmit independent data streams to multi-antenna user terminals, \textcolor{black}{thereby achieving higher throughput.} We first adopt a multi-satellite weighted minimum mean square error (WMMSE) formulation under statistical channel state information (CSI), which yields closed-form updates for the precoding and receive vectors. To overcome the high complexity of optimization, we propose a learning-based WMMSE design that integrates tensor equivariance with closed-form recovery, enabling inference with near-optimal performance without iterative updates. Moreover, to reduce inter-satellite signaling overhead incurred by exchanging CSI and precoding vectors in centralized coordination, we develop a decentralized multi-satellite transmission scheme in which each satellite locally infers its precoders rather than receiving from the central satellite. The proposed decentralized scheme leverages periodically available satellite state information, such as orbital positions and satellite attitude, which is inherently accessible in satellite networks, and employs a dual-branch tensor-equivariant network to predict the precoders at each satellite locally. Numerical results demonstrate that the proposed multi-satellite transmission significantly outperforms single-satellite systems \textcolor{black}{in sum rate}; the decentralized scheme achieves \textcolor{black}{sum-rate} performance close to the centralized schemes while substantially reducing computational complexity and inter-satellite overhead; and the learning-based schemes exhibit strong robustness and scalability across different scenarios. 
\end{abstract}

  \begin{comment}  \begin{IEEEkeywords}
Multi-satellite massive MIMO, learning-based WMMSE, tensor equivariance, decentralized transmission, inter-satellite overhead.
	\end{IEEEkeywords}\end{comment}

	% For peer-reviewed papers, you can put extra information on the cover
	% \fi
	%
	% For review papers, this IEEEtran command inserts a page break and
	\IEEEpeerreviewmaketitle

\vspace{-5mm}\section{Introduction}
\IEEEPARstart{S}{atellite} communications are widely regarded as a key enabling technology for future sixth-generation (6G) wireless networks \cite{WANG202542}. By complementing terrestrial infrastructures, satellite systems can provide seamless global coverage, extend network coverage to remote, oceanic, and underserved areas where terrestrial deployments are costly \cite{Perez2019}. Therefore, satellite communication plays an important role in achieving ubiquitous connectivity, which is one of the central visions of IMT-2030 \cite{cao2025, cao2025_vtc, Ntontin2025}. 

The demand for global coverage and high aggregate capacity has led to the emergence of unprecedentedly large low-earth orbit (LEO) mega-constellations, facilitated by satellite miniaturization and reusable launch vehicles that substantially reduce deployment costs \cite{Christie2018, Gill2023, Liu2023}. Meanwhile, satellite networks are evolving not only toward denser constellations, but also toward larger antenna apertures and more capable user terminals (UTs). Recent platforms such as BlueWalker-3 and Starlink's V2 Mini indicate a trend toward satellites equipped with large-scale antenna arrays for high-gain beamforming \cite{Ntontin2024}, while UTs are increasingly adopting \textcolor{black}{steerable multi-antenna architectures with high spatial directivity and gain} \cite{Amendola2023, WANG202542}. Together, these trends pave the way for multi-satellite massive multiple-input multiple-output (MIMO), where large-scale antenna arrays at satellites and multiple antennas at UTs jointly enhance spatial multiplexing capabilities, improve satellite link budget, and manage interference.

Precoding has emerged as one of the most effective techniques for such massive MIMO satellite communication systems, particularly for improving spectral efficiency, mitigating interference, and enhancing the overall user experience \cite{Khammassi2024}. By leveraging advanced satellite payloads and large-scale antenna arrays, precoding in orthogonal frequency-division multiplexing (OFDM) satellite systems is applied on a per-subcarrier basis to fully exploit spatial multiplexing gains, enabling the reuse of the same time-frequency resources while efficiently serving \textcolor{black}{many} users.
Nevertheless, single-satellite transmission in LEO constellations is fundamentally constrained by limited payload and severe inter-satellite interference arising from dense deployment and aggressive spectrum reuse. These challenges motivate the adoption of multi-satellite cooperative transmission to manage joint user service and interference \cite{WANG_2026MSCT, Shang2026}.
Moreover, the present deployment of multi-satellite precoding often relies on centralized coordination, which requires extensive inter-satellite information exchange through inter-satellite links (ISLs). However, establishing and maintaining reliable ISLs can be challenging in large LEO constellations. In particular, satellites in extreme scenarios, such as those located in counter-rotating orbital seams or different orbital planes, may experience large relative motion, making ISL establishment challenging and causing link instability \cite{Chu2025, Beatriz2019}. Consequently, developing highly efficient and scalable precoding schemes for decentralized multi-satellite transmission that reduce reliance on ISLs becomes increasingly important for practical multi-satellite massive MIMO systems.
\vspace{-4mm}\subsection{Related Works}
Extensive studies have investigated precoding for single-satellite systems. Conventional schemes, including MMSE \cite{Peel2005}, WMMSE \cite{Christensen2008}, and energy-efficient designs \cite{Wu2023}, generally require perfect instantaneous CSI (iCSI), which is difficult to obtain in satellite systems due to long propagation delays and high mobility. This has motivated the development of statistical CSI (sCSI)-based precoding methods. For instance, \cite{You2020} exploited angular and power statistics, \cite{Liu2022} addressed angular uncertainty with per-antenna power constraints, and \cite{Wu2023} proposed a low-complexity precoding-update algorithm to combat temporal channel dynamics. Learning-based precoding has also emerged as a promising direction for improving performance, efficiency, and adaptability \cite{Fontanesi2026}. Representative works include learning simplified scalar parameters for low-complexity precoding \cite{Li2022} and hybrid CNN-LSTM architectures for improved robustness to Doppler shifts and delay distortions \cite{Ying2024}. \cite{Gao2022} mapped implicit CSI to precoding, avoiding explicit channel reconstruction and reducing pilot overhead.

As satellite communication systems evolve toward dense LEO networks with increasing ISL deployment, cooperative multi-satellite transmission is becoming essential \cite{Shang2026, Abdelsadek2022, Wu2025, WANG_2026MSCT}. By allowing multiple satellites to jointly serve users, such architectures can enhance link performance, suppress inter-satellite interference, and improve spatial resource utilization \cite{Bakhsh2025, Chen2024}. Existing studies have mainly focused on centralized architectures \cite{Xiang2024, wang2026}, where a central satellite aggregates multi-satellite CSI and jointly optimizes precoding vectors. \cite{Wu2025} investigated cooperative downlink transmission via distributed beamforming among multiple LEO satellites, analyzing the impact of inevitable delay and Doppler compensation errors on coherent processing. Building on this line of research, \cite{wang2026} proposed a scalable transformer-based precoding scheme for coherent transmission that learns from sCSI and can adapt to varying numbers of users and satellites. More recently, \cite{wang2025MSMS} extended multi-satellite coherent transmission to the multi-antenna UT setting with multi-stream transmission. In contrast, multi-satellite non-coherent (NCT) transmission was only preliminarily studied in \cite{Xiang2024}. Owing to its much lower synchronization requirements, NCT transmission constitutes a potentially more practical and deployable solution for dynamic LEO networks \cite{Shang2026}. However, existing designs still mainly rely on centralized optimization and high-dimensional CSI exchange, which incur substantial computational complexity and inter-satellite overhead, thereby limiting their practical applicability.

To overcome the limitations of centralized architectures, decentralized multi-satellite transmission has attracted increasing attention. In such architectures, each satellite performs local precoding while exchanging only limited state information, thus substantially reducing inter-satellite communication overhead \cite{Zhang2024}. Although this decentralized operation may incur some performance loss relative to fully centralized optimization, it is often more practical in satellite networks due to inter-satellite and satellite–terrestrial delays, as well as limited onboard resources \cite{Rodrigues2023}.
Recent progress includes privacy-preserving split learning \cite{Sun2024}, multi-agent reinforcement learning for cooperative scheduling \cite{Martinez2025}, and decentralized federated learning for adaptive model exchange and local training \cite{Cheng2025}. Nevertheless, existing studies mainly address distributed learning and resource management, with limited attention to physical-layer design for multi-satellite massive MIMO systems. In particular, decentralized transmission strategies built on consensus mechanisms within a unified operational framework remain largely unexplored.
\vspace{-4mm}\subsection{Contributions}
Despite the promise of multi-satellite massive MIMO transmission in improving link performance and suppressing inter-satellite interference \cite{Abdelsadek2022,Bakhsh2025}, several important gaps remain. Existing learning-based designs have been limited to systems with single-antenna UTs \cite{wang2026}, while the distinctive properties of multi-antenna UT scenarios have not been fully exploited \cite{Xiang2024,Zhang2026}. Meanwhile, multi-satellite NCT transmission, in which multiple satellites transmit independent data streams to the same multi-antenna UT, offers greater practical potential owing to its lower synchronization requirements \cite{Shang2026,wang2026,wang2025MSMS}, yet still suffers from the high computational complexity of conventional iterative schemes \cite{Xiang2024}. Furthermore, centralized multi-satellite transmission is fundamentally constrained by limited ISL communication, including stringent requirements on information timeliness and restricted exchange capacity.
These constraints are further aggravated by the difficulty of acquiring CSI and the overhead associated with precoder dissemination \cite{Chu2025,Shang2026}. As a result, decentralized learning-based NCT transmission for multi-satellite systems remains largely unexplored.
These observations motivate the following fundamental research question: \textit{\textcolor{black}{How to design centralized or decentralized NCT transmission schemes for cooperative multi-satellite systems by leveraging deep learning techniques?}} Addressing this question constitutes the central focus of this paper. In summary, the main contributions of this work are as follows:

\begin{itemize}  
\begin{comment}\item 
For centralized multi-satellite massive MIMO transmission, we first formulate the weighted sum rate (WSR) maximization problem and reformulate it as a multi-satellite WMMSE problem. Under per-satellite power constraints, we derive closed-form iterative updates for the transmit precoding and receive vectors, and the coupled power constraints are reduced to a low-dimensional Lagrange multiplier optimization, enabling efficient computation. Then, we design a centralized alternating optimization algorithm that serves both as an optimization-based performance scheme and as the basis for the subsequent learning-based approaches.
\end{comment}
\item To mitigate the high complexity of iterative optimization, we extract physically interpretable low-dimensional variables from the closed-form WMMSE solutions. Leveraging this structure, we propose a centralized learning-based WMMSE framework that incorporates a tensor equivariant network designed to capture and preserve the permutation-equivariant structure across both satellite and UT dimensions. The low-dimensional variables are predicted and then mapped back to the precoders through a closed-form recovery. As a result, the proposed scheme achieves near-WMMSE performance with substantially lower online inference complexity while naturally adapting to different numbers of satellites and UTs.

\item To address the excessive inter-satellite overhead and latency of centralized coordination, we first establish the feasibility of decentralized multi-satellite transmission from the structure of the closed-form WMMSE solution. Specifically, we show that each satellite can predict its local precoding vectors primarily from local CSI, while the cross-satellite coupling can be captured through periodically exchanged low-rate satellite state information (SSI), including orbital positions and satellite attitude. Motivated by this observation, we further propose a decentralized multi-satellite transmission architecture that eliminates the need for centralized CSI aggregation and precoder dissemination, thereby significantly reducing inter-satellite communication overhead.

\item Building on the proposed decentralized architecture, we further develop a learning-based decentralized WMMSE scheme. Specifically, a dual-branch tensor equivariant network is designed to process local-satellite and other-satellite features separately, while preserving the permutation-equivariant and permutation-invariant structures inherent in the satellite and UT dimensions of the decentralized inference problem. The network predicts low-dimensional variables in the closed-form WMMSE expressions, which are then recovered into the corresponding precoding and receive vectors. As a result, the proposed decentralized scheme enables low-complexity online inference, parallel execution across satellites, and scalability to varying numbers of satellites and UTs, making it more suitable for practical LEO satellite networks.
\end{itemize}
	
This paper is structured as follows: \secref{System} presents the multi-satellite system model and problem formulation. \secref{centralized_network} introduces the centralized learning-based WMMSE transmission design. \secref{decentralized_arch} presents the decentralized multi-satellite transmission architecture. \secref{decentralized_network} develops the corresponding learning-based decentralized WMMSE network design. \secref{simulation} provides the simulation results. Finally, \secref{conclusion} concludes this article.

{\textit{Notation}}: $x, {\bf x}, {\bf X}$ represent scalar, column vector, matrix. $(\cdot)^T$, $(\cdot)^{*}$, $(\cdot)^H$, and $(\cdot)^{-1}$ respectively denote the transpose, conjugate, transpose-conjugate, and inverse operations. $\left \|\cdot\right \|_{2}$ denotes $l_2$-norm. $\otimes$ and $\odot$ are the Kronecker product and Hardmard product operations. The operator ${\rm Tr}\left\{\cdot\right\}$ represents the matrix trace. $\mathbb{E}\left\{\cdot\right\}$ denotes the expectation. The expression $\mathcal{C}\mathcal{N}(\mu, \sigma^2)$ denotes circularly symmetric Gaussian distribution with expectation $\mu$ and variance $\sigma^2$. ${\mathbb{R}}^{M\times N}$ and ${\mathbb{C}}^{M\times N}$ represent the set of $M\times N$ dimension real- and complex-valued matrixes. $\nabla f$ denotes gradient of function $f(\cdot)$. $k\in \mathcal{K}$ means element $k$ belongs to set $\mathcal{K}$. $\Re(\cdot) $ and $\Im(\cdot)$ denote the real and imaginary parts of a complex scalar, vector, or matrix.
We use  $\tX_{[m_{1}, \ldots, m_{N}]}$ to denote the indexing of elements in tensor  $\tX \in \mathbb{R}^{M_{1} \times \cdots \times M_{N}} .[\tA_{1}, \ldots, \tA_{K}]_{S}$  denotes the tensor formed by stacking  $\tA_{1}, \ldots, \tA_{K}$  along the  $S$-th dimension.  $[\cdot]_{0}$ denotes the concatenation of tensors along a new dimension, i.e.,  $\tA_{[n,:, \cdots]}=\tB_{n}, \forall n$  when $\tA=\left[\tB_{1}, \ldots, \tB_{N}\right]_{0}$. We define the product of tensor $\tX \in \mathbb{R}^{M_{1} \times \cdots \times M_{N} \times D_{X}}$ and matrix $\mathbf{Y} \in \mathbb{R}^{D_{X} \times D_{Y}}$  as  $(\tX \times \mathbf{Y})_{[m_{1}, \ldots, m_{N},:]}= \tX_{[m_{1}, \ldots, m_{N},:]} \mathbf{Y}$.
\section{Multi-Satellite NCT Transmission System Model and Problem Formulation}\label{System}
We consider a downlink (DL) multi-satellite cooperative transmission system, as shown in \figref{MSIT_scenario}, comprising $S$ LEO satellites, performing transmission to $K$ \textcolor{black}{mobile} UTs within the same time-frequency resources. For clarity, we define the satellite set as $\mathcal{S} = \{1, \ldots, S\}$ and the UT set as $\mathcal{K} = \{1, \ldots, K\}$. Each satellite is equipped with a large-scale uniform planar array (UPA) consisting of  $M = M_{\text{x}} \times M_{\text{y}}$  antenna elements, where  $M_{\text{x}}$  and  $M_{\text{y}}$  represent the number of antennas along the x-axis and y-axis, respectively. Similarly, the UTs are equipped with a UPA composed of  $N = N_{\text{x}^{\prime}} \times N_{\text{y}^{\prime}}$  antenna elements.
In this cooperative transmission framework, NCT transmission is enabled by multi-antenna UT reception, allowing distributed precoding across satellites to deliver independent data streams to the UTs. Compared with schemes where multiple satellites serve the same UTs over orthogonal time-frequency resources, cooperative precoding on shared resources can better exploit spatial division multiplexing gains, \textcolor{black}{leading to enhanced spectral efficiency and improved resource utilization.}
\begin{figure}[tp]
	\centering
    \vspace{-8mm}
\includegraphics[width=0.8\linewidth,trim=0.1cm 0.8cm 0.1cm 0.5cm,clip]{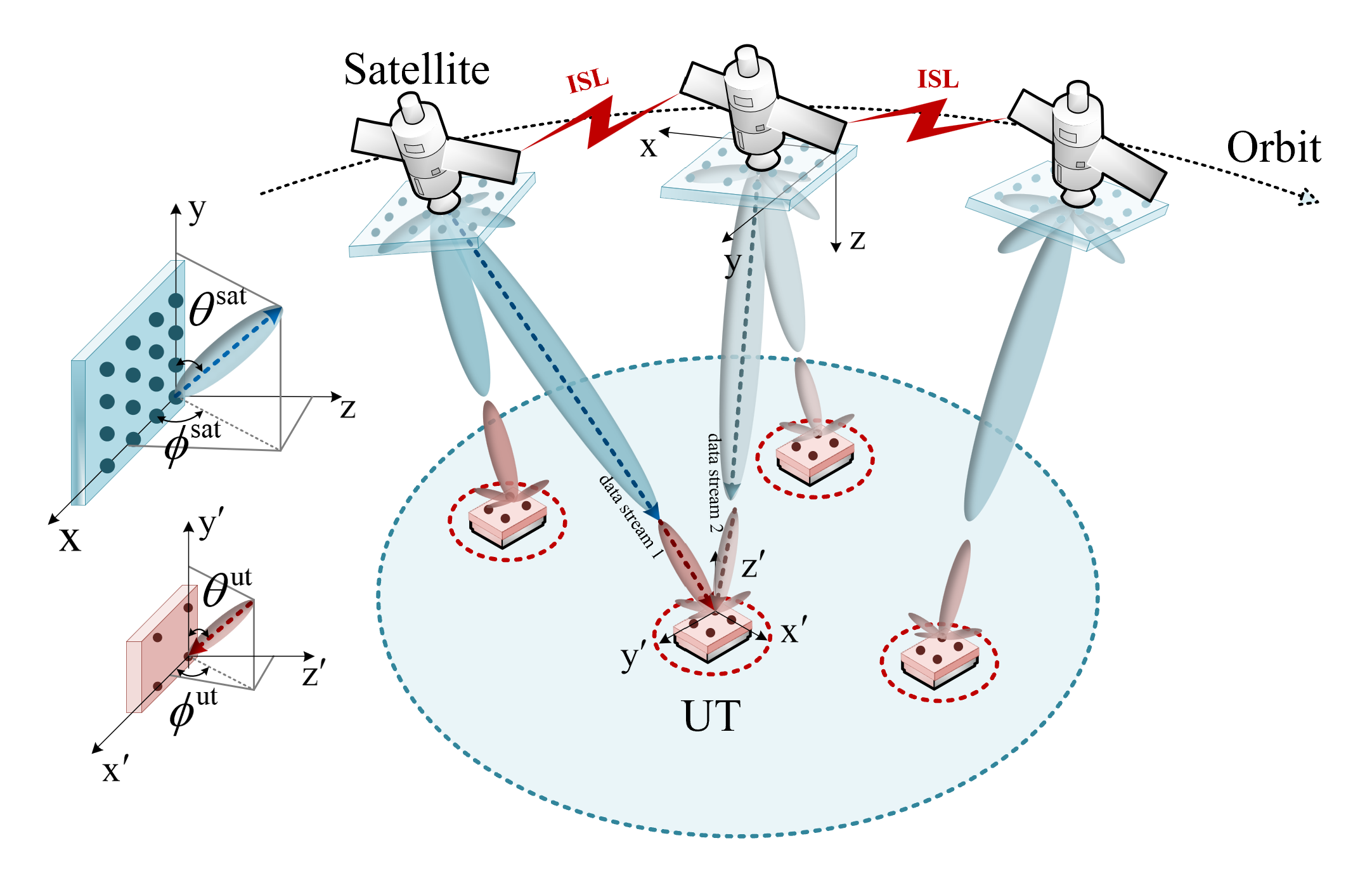}
\vspace{-1mm}\caption{\textcolor{black}{The system architecture of the multi-satellite transmission.}}
\label{MSIT_scenario}\vspace{-5mm}
\end{figure}
\vspace{-3mm}\subsection{Channel and Signal Models}
To support broadband communication in the multi-satellite system, the entire system bandwidth is employed for OFDM transmission. Let the number of subcarriers be denoted by \(N_{\rm sc}\), the cyclic prefix (CP) length by \(N_{\rm cp}\), and the system sampling period by \(T_{\rm sam}\). Accordingly, the CP duration is given by \(T_{\rm cp}=N_{\rm cp}T_{\rm sam}\). The OFDM symbol duration without CP is \(T_{\rm sc}=N_{\rm sc}T_{\rm sam}\), while the total symbol duration including CP is \(T_{\rm sym}=T_{\rm sc}+T_{\rm cp}\). The time-domain transmit signal vector of the \(s\)-th satellite in the \(m\)-th OFDM symbol, denoted as \(\mathbf{x}_{s,m}(t)\in \mathbb{C}^{M \times 1}\), can be expressed as:
\begin{align}
\mathbf{x}_{s,m}(t)\!=\!\!\! \sum_{r=0}^{N_{\rm sc}-1} \!\!\mathbf{x}_{s,m,r} {\mathrm e}^{j2\pi r \Delta f t},  - T_{\rm cp} \!\leq\! t\!-\!mT_{\rm sym}\! < \!T_{\rm sc},
\end{align}
herein, $\mathbf{x}_{s,m,r}$ denotes the frequency-domain transmit signal vector over the $r$-th subcarrier in the $m$-th OFDM symbol and $\Delta f =1/T_{\rm sc}$. The DL time-domain received signal vector of the $k$-th UT from the $s$-th satellite in the $m$-th OFDM symbol at the $t$-th time instant can be written as:
\begin{align}
\mathbf{y}_{s,k,m}(t)=\int_{-\infty}^{\infty} \tilde{\mathbf{H}}_{s,k}(t,\tau)\mathbf{x}_{s,m}(t-\tau){\rm d}\tau+\mathbf{z}_{s,k,m}(t),
\end{align}
\textcolor{black}{where we denote $\tau$ as the delay variable,} $\mathbf{z}_{s,k,m}(t)=\sum_{t \in \mathcal{S} \setminus s} \int_{-\infty}^{\infty} \tilde{\mathbf{H}}_{t,k}(t, \tau) \mathbf{x}_{t,m}(t-\tau)+ \tilde{\mathbf{n}}_{k}(t)$, and $\tilde{\mathbf{n}}_{k}(t) \in \mathbb{C}^{N \times 1}$ denotes
the additive white Gaussian noise vector with distribution $\mathcal{C N}(\mathbf{0}, \sigma_{k}^{2} \mathbf{I}_N)$.
Herein, the time-varying spatial
domain MIMO channel $\tilde{\mathbf{H}}_{s,k}(t,\tau) \in \mathbb{C}^{N \times M}$ from the $s$-th satellite to the $k$-th UT can be expressed as \cite{Li2022}:
\begin{align}
\tilde{\mathbf{H}}_{s,k}(t,\tau)\!=\!\!\!\!\!\sum_{\ell=0}^{L_{s,k}-1} \!\!a_{s,k, \ell} {\mathrm e}^{j 2 \pi \nu_{s,k,\ell} t} \delta(\tau\!-\!\tau_{s,k,\ell})\mathbf{d}_{s,k, \ell}\mathbf{g}_{s,k,\ell}^H,
\end{align}
where $L_{s,k}$ is the multipath number of the channel from the $s$-th satellite to the $k$-th UT, $a_{s,k,\ell}$, $\nu_{s,k,\ell}$, $\tau_{s,k,\ell}$, $\mathbf{d}_{s,k,\ell}$ and $\mathbf{g}_{s,k,\ell}$ are the complex channel gain, Doppler shift, propagation delay, receive array response vector and transmit array response vector associated with the $\ell$-th path.

The satellite position, attitude information, and the global navigation satellite system (GNSS) of the UTs can be jointly exploited to estimate the angular parameters associated with the LoS path. Specifically, the angle-of-departure (AoD) at the satellite side is denoted by $(\theta_{s,k,\ell}^{\rm sat}, \phi_{s,k,\ell}^{\rm sat})$, while the angle-of-arrival (AoA) at the UT side is denoted by $(\theta_{s,k,\ell}^{\rm ut}, \phi_{s,k,\ell}^{\rm ut})$.
In LEO satellite channels, owing to the large satellite-to-ground distance and the fact that scattering occurs within a few kilometers around the UT, the AoD pairs corresponding to different multipath components of the same link are nearly identical. Therefore, we have
$(\theta_{s,k,\ell}^{\rm sat}, \phi_{s,k,\ell}^{\rm sat}) = (\theta_{s,k}^{\rm sat}, \phi_{s,k}^{\rm sat}), \forall \ell$.
Accordingly, the transmit array response vectors at the satellite side can be expressed as \cite{You2020,Liu2022}:
\begin{align}
\mathbf{g}_{s,k,\ell}\! =\! \mathbf{g}_{s,k}\! =\! \mathbf{a}_{N_{\rm x}}\!\!\left(\sin \theta^{\rm sat}_{s,k} \cos \phi^{\rm sat}_{s,k}\right)\! \otimes \!\mathbf{a}_{N_{\rm y}}\!\!\left(\cos \theta^{\rm sat}_{s,k}\right).
\end{align}
The receive array response vectors are given by:
\begin{align}
\mathbf{d}_{s,k,\ell} = \mathbf{a}_{N_{\rm x^{\prime}}}\left(\sin \theta^{\rm ut}_{s,k,\ell} \cos \phi^{\rm ut}_{s,k,\ell}\right) \otimes \mathbf{a}_{N_{\rm y^{\prime}}}\left(\cos \theta^{\rm ut}_{s,k,\ell}\right).
\end{align}
We define $\mathbf{a}_{n_{\rm v}}(x) = \frac{1}{\sqrt{n_{\rm v}}}[1, \mathrm{e}^{-\frac{j 2 \pi}{\lambda}d_{\rm v}x}, \ldots, \mathrm{e}^{-\frac{j 2 \pi}{\lambda} d_{\rm v}(n_{\rm v}-1)x}]^{T}$. \( d_{\rm v} = \frac{\lambda}{2} \) represents the distance between adjacent antenna elements along the $\rm v$-axis and \( {\rm v}\in \left\{\rm x,\rm y,\rm x^{\prime},\rm y^{\prime} \right\} \).
Since the signals from different satellites arrive asynchronously at the UT, joint processing is challenging. Thus, spatial linear receive processing is adopted to separate the signals, the resulting received signal $r_{s,k,m}(t)$ is given by:
\begin{align}
&r_{s,k,m}(t)=\mathbf{b}_{s,k}^{H} \mathbf{y}_{s,k,m}(t)\nonumber\\
&=\mathbf{b}_{s,k}^{H} \left(\int_{-\infty}^{\infty} \tilde{\mathbf{H}}_{s,k}(t,\tau)\mathbf{x}_{s,m}(t-\tau){\rm d}\tau+\mathbf{z}_{s,k,m}(t)\right),
\end{align}
\( \mathbf{b}_{s,k} \!\!\in\!\! \mathbb{C}^{N \times 1}\) denotes the receive vector at the $k$-th UT for extracting the signal from the $s$-th satellite. \textcolor{black}{The \(k\)-th UT estimates the Doppler shift and propagation delay associated with the signal from the \(s\)-th satellite using the satellite and UT positions obtained from ephemeris and GNSS \cite{Xiang2024}.} After UT performing Doppler and delay compensation, the compensated signal at the \(k\)-th UT from the \(s\)-th satellite is:
\begin{align}
&r_{s,k,m}^{\rm cps}(t)={r}_{s,k,m}(t+\tau_{s,k}^{\rm cps})\mathrm{e}^{-j2\pi\nu_{s,k}^{\rm cps}(t+\tau_{s,k}^{\rm cps})}\nonumber\\
&=\mathbf{b}_{s, k}^{H}\left(\int_{-\infty}^{\infty} \hat{\mathbf{H}}_{s, k}^{\mathrm{cps}}(t, \tau) \mathbf{x}_{s, m}(t-\tau) \mathrm{d} \tau+\mathbf{z}_{s, k, m}^{\mathrm{cps}}(t)\right),
\end{align}
where $\nu_{s,k}^{\rm cps}$ and $\tau_{s,k}^{\rm cps}$
are compensated Doppler and delay, respectively. The compensated equivalent channel is:
\begin{align}
\hat{\mathbf{H}}_{s, k}^{\mathrm{cps}}(t, \tau)\!=\!\!\!\!\sum_{\ell=0}^{L_{s, k}-1} \!\!\!\hat{a}_{s, k, \ell} \mathrm{e}^{j 2 \pi \nu_{s, k, \ell}^{\mathrm{ut}} t} \delta\!\left(\tau-\tau_{s, k, \ell}^{\mathrm{ut}}\right) \mathbf{d}_{s, k, \ell} \mathbf{g}_{s, k}^{H}.
\end{align}
 Herein, \( \hat{a}_{s,k,\ell} \triangleq a_{s,k,\ell} \mathrm{e}^{j 2 \pi \nu_{s,k,\ell}^{\mathrm{ut}} \tau_{s,k}^{\mathrm{cps}}} \). The residual Doppler shift and propagation delay are defined as \( \nu_{s,k,\ell}^{\mathrm{ut}} \triangleq \nu_{s,k,\ell}-\nu_{s,k}^{\mathrm{cps}} \) and
\( \tau_{s,k,\ell}^{\mathrm{ut}} \triangleq \tau_{s,k,\ell}-\tau_{s,k}^{\mathrm{cps}} \), respectively.
The compensated $\mathbf{z}_{s,k,m}^{\mathrm{cps}}(t)$ is given by \( \mathbf{z}_{s,k,m}^{\mathrm{cps}}(t) \triangleq \mathbf{z}_{s,k,m}(t+\tau_{s,k}^{\mathrm{cps}})\mathrm{e}^{-j 2 \pi \nu_{s,k}^{\mathrm{cps}}(t+\tau_{s,k}^{\mathrm{cps}})} \). \textcolor{black}{Since the Doppler and delay effects caused by satellite mobility and long propagation distances are properly compensated, the residual Doppler and delay are sufficiently small to support reliable OFDM transmission \cite{Xiang2024}.} Consequently, time and frequency synchronization can be assumed for the link between the \(s\)-th satellite and the \(k\)-th UT after compensation, whereas the signals from other satellites remain asynchronous at the \(k\)-th UT. After the above processing, we assume that the channel is invariant within one OFDM symbol, and the channel may change from symbol to symbol. After OFDM demodulation, the received signal over the $r$-th subcarrier in the $m$-th OFDM symbol is given by \cite{Li2022,Xiang2024}:
\begin{align}
\begin{split}
r_{s,k,m,r} = \mathbf{b}_{s,k}^H\mathbf{H}_{s,k,m,r}^{\rm eff}\mathbf{x}_{s,m}+{z}_{s,k,m,r},
\end{split}
\end{align}
Without ambiguity, we omit the subscripts of OFDM symbol $m$ and subcarrier $r$.
The effective time-frequency domain channel matrix between the $s$-th satellite and the $k$-th UT $\mathbf{H}_{s,k}^{\rm eff}= \mathbf{d}_{s, k}^{\rm eff} \mathbf{g}_{s,k}^{H} \in \mathbb C^{N\times M}$.
The transmit signal $\mathbf{x}_{s,m}=\mathbf{x}_s$ is:
\begin{align}
\mathbf{x}_{s}=\sum_{k\in \mathcal{K}} \mathbf{p}_{s,k}{x}_{s,k} \in \mathbb C^{M\times 1},
\end{align}
herein, $\mathbf{p}_{s,k} \in \mathbb{C}^{M \times 1}$ denotes the precoding vector of the $s$-th satellite. Then, we have the received signal:
\begin{align}
r_{s,k}=\mathbf{b}_{s,k}^{H} \mathbf{H}_{s,k} \sum_{m \in \mathcal K}\mathbf{p}_{s,m}x_{s,m}+z_{s,k},
\end{align}
\textcolor{black}{where \(z_{s,k}\) denotes the aggregate interference-plus-noise term, including both inter-satellite and intra-satellite interference, at the \(k\)-th UT for the \(s\)-th satellite,} whose variance is given by  $\mathbb{E}\{z_{s,k} z_{s,k}^*\} = \mathbf{b}^H_{s,k}\mathbf{R}_{s,k}\mathbf{b}_{s,k}$, and the covariance matrix is:
\begin{align}
\mathbf{R}_{s,k}=\sum_{t \in \mathcal{S} \setminus s} \sum_{m \in \mathcal{K}} \mathbf{g}_{t,k}^{H} \mathbf{p}_{t,m} \mathbf{p}_{t,m}^H \mathbf{g}_{t,k} \mathbf{R}_{t,k}^{\mathrm{ut}}+\sigma_{k}^{2} \mathbf{I}_{N},
\end{align}
where $\mathbf{R}_{t,k}^{\mathrm{ut}}$ denotes the channel correlation matrix at UT side.
Note that we consider multi-satellite NCT transmission, where multiple satellites transmit independent data streams to the multi-antenna UT, which differs from coherent multi-satellite transmission, where identical data streams are transmitted and coherently combined at the UT \cite{wang2026,wang2025MSMS}. 
\subsection{Statistical Characteristics of Satellite Channel}
The channel between the $s$-th satellite and the $k$-th UT can be further modeled as a Rician fading channel \cite{Xiang2024,Li2022,wang2025MSMS}:
\begin{align}
\begin{split}
\mathbf{H}_{s,k}=\sqrt{\frac{\kappa_{s,k} \beta_{s,k}}{\kappa_{s,k}+1}} \mathbf{d}_{s,k,0}\mathbf{g}_{s,k}^{H}+\sqrt{\frac{\beta_{s,k}}{\kappa_{s,k}+1}} \tilde{\mathbf{d}}_{s,k}\mathbf{g}_{s,k}^{H},
\end{split}
\end{align}
herein, the LoS component satisfies $\|\mathbf{d}_{s,k,0}\|^{2}=\|\mathbf{g}_{s,k}\|^{2}=1$, while the NLoS component follows a complex Gaussian distribution $\tilde{\mathbf{d}}_{s,k} \sim \mathcal{CN}(\mathbf{0},\boldsymbol{\Sigma}_{s,k})$ with $\operatorname{Tr}(\boldsymbol{\Sigma}_{s,k})=1$. 
The receiver steering vector $\mathbf{d}_{s,k}=\sqrt{\frac{\kappa_{s,k}\beta_{s,k}}{\kappa_{s,k}+1}}\mathbf{d}_{s,k,0} 
+\sqrt{\frac{\beta_{s,k}}{\kappa_{s,k}+1}}\tilde{\mathbf{d}}_{s,k}$,
and $\beta_{s,k}=\textcolor{black}{\mathbb{E}_{\mathbf{d}}}\{\operatorname{Tr}(\mathbf{H}_{s,k}\mathbf{H}_{s,k}^{H})\} 
=\textcolor{black}{\mathbb{E}_{\mathbf{d}}}\{\|\mathbf{d}_{s,k}\|^{2}\}$ denotes the average channel power. The channel correlation matrices at the UT side can be calculated as:
\begin{align}
\!\!\!\!\!\!\!\!\!\mathbf{R}_{s,k}^{\rm {ut}} \!\!=\!\!\textcolor{black}{\mathbb{E}_{\mathbf{d}}}\!\!\left\{\!\mathbf{H}_{s,k} \mathbf{H}_{s,k}^{H}\!\right\}\!\!=\!\!\frac{\kappa_{s,k} \beta_{s,k}}{\kappa_{s,k}\!+\!1} \mathbf{d}_{s,k,0} \mathbf{d}_{s,k, 0}^{H}\!\!+\!\!\frac{\beta_{s,k}}{\kappa_{s,k}\!+\!1} \!\boldsymbol{\Sigma}_{s,k},\label{Rut}
\\
\mathbf{R}_{s,k}^{\rm {sat}}  =\textcolor{black}{\mathbb{E}_{\mathbf{d}}}\left\{\mathbf{H}_{s,k}^{H} \mathbf{H}_{s,k}\right\}=\beta_{s,k}\mathbf{g}_{s,k} \mathbf{g}_{s,k}^{H}.\label{Rsat}
\end{align}
\textcolor{black}{Based on the satellite channel characteristics described above, we summarize the sCSI available at the satellite side: $\{\phi^{\rm sat}_{s,k}, \theta^{\rm sat}_{s,k}, \phi^{\rm ut}_{s,k,0}, \theta^{\rm ut}_{s,k,0}, \boldsymbol{\Sigma}_{s,k}, \,\beta_{s,k}, \kappa_{s,k}\}_{\forall s,\forall k}$.}
The DL ergodic rate from the $s$-th satellite to the $k$-th UT is:
\begin{align}
R_{s,k}&=\mathbb{E}_{\mathbf{H}}\left\{\log_2 \left(1+{{\alpha}_{s,k}^{-1}\mathbf{b}_{s,k}^{H} \mathbf{H}_{s,k} \mathbf{p}_{s,k}\mathbf{p}_{s,k}^H\mathbf{H}_{s,k}^H\mathbf{b}_{s,k}}\right) \right\}\nonumber\\
&=\textcolor{black}{\mathbb{E}_{\mathbf{d}}}\left\{\log_2 \left(1+{\alpha_{s,k}^{-1} |\mathbf{p}_{s,k}^H\mathbf{g}_{s,k}|^2|\mathbf{b}_{s,k}^H\mathbf{d}_{s,k}|^2}\right) \right\},
\end{align}
herein, ${\alpha}_{s,k}\!\triangleq\! \sum_{m \in \mathcal{K}\setminus k}|\mathbf{p}_{s,m}^H\mathbf{g}_{s,k}|^2|\mathbf{b}_{s,k}^H\mathbf{d}_{s,k}|^2\!\!+z_{s,k}z_{s,k}^*$.
\vspace{-2mm}\subsection{Multi-Satellite WSR Maximization Problem Formulation}\label{centralized_optimization}
\textcolor{black}{Assuming the sCSI available at satellites,} the WSR maximization problem for multi-satellite transmission can be formulated as \cite{Xiang2024}:
\begin{align}
	\begin{split}
	\max_{\left\{\mathbf{b}_{s,k},\mathbf{p}_{s,k}\right\}_{\forall s, k}}&\quad \sum_{s \in \mathcal S} \sum_{k \in \mathcal{K}}  a_{s,k}R_{s,k}, \\
		\text{ s.t. } 
		& \quad \sum_{k \in \mathcal{K}} \|\mathbf{p}_{s,k}\|_2^2 \leq P_s^{\rm sat},  \forall s \in  \mathcal S,
	\end{split}
\end{align}
$a_{s,k}$ denotes the user rate weight, and $P_s^{\rm sat}$ is the satellite transmit power budget. To facilitate the optimization, we optimize the WSR maximization problem by solving the following WMMSE problem:
\begin{align}
	\begin{split}
	\!\!\!\underset{\left\{{u}_{s,k},{w}_{s,k},\mathbf{p}_{s,k},\mathbf{b}_{s,k}\right\}_{\forall s, k}}{\min} \quad &\sum_{s \in \mathcal S} \sum_{k \in \mathcal{K}} \left( w_{s,k}e_{s,k}-\log w_{s,k}\right), \\
	\!\!\!	\text{ s.t. } 
		 \quad &\sum_{k \in \mathcal{K}} \|\mathbf{p}_{s,k}\|_2^2 \leq P_s^{\rm sat},  \forall s \in  \mathcal S,
	\end{split}
\end{align}
where $w_{s,k}$ represents the error weight, and the MSE $e_{s,k}$ can be expressed as:
\begin{align}
	\begin{split}
e_{s,k}&=\mathbb{E}_{\mathbf{H},{x},{z}}\left\{|u_{s,k}r_{s,k}-x_{s,k}|^2\right\}\\
&=\left|u_{s, k}\right|^{2} \left(\eta_{s, k}+\zeta_{s,k}\right)-u_{s, k}^*\xi_{s, k}^*-u_{s, k}\xi_{s, k}+1,
 \end{split}
\end{align}
where the desired signal level $\xi_{s,k} =  \mathbb{E}_\mathbf{H}\{\mathbf{b}_{s, k}^{H} \mathbf{H}_{s, k} \mathbf{p}_{s, k}\}=\sqrt{\frac{\kappa_{s,k}\beta_{s,k}}{\kappa_{s,k}+1}}\mathbf{b}_{s, k}^{H} \mathbf{d}_{s,k,0}\mathbf{g}_{s,k}^H \mathbf{p}_{s,k}$, and the desired signal power $\zeta_{s,k}=\mathbb{E}_\mathbf{H}\{|\mathbf{b}_{s, k}^{H} \mathbf{H}_{s, k} \mathbf{p}_{s, k}|^2\}=|\mathbf{p}_{s,k}^H\mathbf{g}_{s,k}|^2\mathbf{b}_{s,k}^H\mathbf{R}_{s,k}^{\rm ut}\mathbf{b}_{s,k}$. Moreover, the aggregate interference-plus-noise signal power $\eta_{s, k}=\mathbb{E}_{\mathbf{H},z}\{\sum_{m\in \mathcal{K}\setminus k}|\mathbf{b}_{s, k}^{H} \mathbf{H}_{s, k} \mathbf{p}_{s, m}|^2+z_{s,k}z_{s,k}^*\}=\sum_{m \in \mathcal{K}\setminus k}|\mathbf{p}_{s,m}^H\mathbf{g}_{s,k}|^2\mathbf{b}_{s,k}^H\mathbf{R}_{s,k}^{\rm ut}\mathbf{b}_{s,k}+\mathbf{b}_{s,k}^H\mathbf{R}_{s,k}\mathbf{b}_{s,k}$.
\begin{algorithm}[t]
\caption{Centralized Optimization-Based WMMSE}
	\label{opt-wmmse-central}
		\begin{algorithmic}[1]
			\STATE \textbf{Input:} \!$\!\{\mathbf{g}_{s,k},\!\mathbf{d}_{s,k,0},\!\beta_{s,k},\!\kappa_{s,k},\!\mathbf{\Sigma}_{s,k}\}_{\forall s, k}$,\!$\{\sigma_{k}^2\}_{\forall k}$,\! $\{P_s^{\rm sat}\}_{\forall s}$.
			\STATE Initialize \(\{{\mathbf{p}_{s,k}}^{(0)},{\mathbf{b}_{s,k}}^{(0)}\}_{\forall s, k}\), $\epsilon^{(0)}$ and $n=1$.
			\STATE \textbf{repeat} 
			\STATE \quad Calculate  $\{{{u}_{s,k}^{\star}}^{(n)},{{w}_{s,k}^{\star}}^{(n)}\}_{\forall s, k}$ from \eqref{mu} and \eqref{weight}. 
			\STATE \quad Solve vector ${\boldsymbol{\lambda}^{\star}}^{(n)}$ from problem \eqref{WMMSE2}.
		    \STATE \quad Obtain $\{{\mathbf{p}^{\star}_{s,k}}^{(n)}\}_{\forall s, k}$ and $\{{\mathbf{b}^{\star}_{s,k}}^{(n)}\}_{\forall s, k}$ from \eqref{precoding_cf}\eqref{beamforming_cf}. 
        	\STATE \quad Calculate $\epsilon^{(n)}$ in \eqref{WMMSE2}.
  			\STATE \quad Set $n = n + 1$.
			\STATE \textbf{until} $n \geq I_{\rm max}^{\rm out}$ \textbf{or} $\epsilon$ convergence.
			\STATE \textbf{Output:} $\{\mathbf{p}_{s,k}^{\star},\mathbf{b}_{s,k}^{\star}\}_{\forall s, k}$.
		\end{algorithmic}
\end{algorithm}

With \(\{\mathbf{b}_{s,k},\mathbf{p}_{s,k}\}_{\forall s, k}\) and \(\{{w}_{s,k}\}_{\forall s, k}\) fixed, the optimal value of \(u_{s,k}^{\star}\) can be determined by setting the derivative of \(e_{s,k}(u_{s,k})\) with respect to \(u_{s,k}^*\) equal to zero, which yields:
\begin{align}
	\begin{split}
   u_{s, k}^{\star}
   =\frac{\xi_{s, k}^*}{\zeta_{s, k}+\eta_{s,k}},\quad
    e_{s, k}^{\star}=1-\frac{|\xi_{s, k}|^2}{\zeta_{s,k}+\eta_{s, k}}.
\label{mu}
	\end{split}	
\end{align}

 By equating the gradients of their respective Lagrangian functions, the optimal MSE-weights \(\{{w}_{s,k}\}_{\forall s, \forall k}\) can be derived for given \(\{\mathbf{b}_{s,k},\mathbf{p}_{s,k}\}_{\forall s,\forall k}\) and \(\{{u}_{s,k}^{\star}\}_{\forall s, \forall k}\) as \cite{Shi2011}:
\begin{align}
	\begin{split}
		{w}_{s,k}^{\star}=\frac{a_{s,k}}{{e}_{s,k}}.
        \label{weight}
	\end{split}	
\end{align}
Then we fix vectors  \(\left\{\mathbf{w}_{s}^{\star}\right\}_{\forall s}\)  and \(\left\{\boldsymbol{u}_{s}^{\star}\right\}_{\forall s}\), the optimization problem is reduced to:
\begin{align}
	\begin{split}
	      \min_{\left\{\mathbf{b}_{s,k},\mathbf{p}_{s,k}\right\}_{\forall s, k}}&\quad \sum_{s \in \mathcal S} \sum_{k \in \mathcal{K}} w_{s,k}e_{s,k}(\left\{\mathbf{b}_{s,k},\mathbf{p}_{s,k}\right\}_{\forall s, k}), \\
		\text{ s.t. } 
		& \quad \sum_{k \in \mathcal{K}} \|\mathbf{p}_{s,k}\|_2^2 \leq P_s^{\rm sat},  \forall s \in  \mathcal S.
	\end{split}	\label{WMMSE1}
\end{align}
The Lagrange function is $\mathcal{L}(\{\{\mathbf{b}_{s,k},\mathbf{p}_{s,k}\}_{\forall k},\lambda_{s}\}_{\forall s})=\sum_{\forall s,\forall k}
w_{s,k}\,
e_{s,k}(\{\mathbf{b}_{s,k},\mathbf{p}_{s,k}\}_{\forall s,k})+\sum_{\forall s}
\lambda_{s}
(
\sum_{k \in \mathcal{K}} \|\mathbf{p}_{s,k}\|_2^{2}
- P_s^{\rm sat}
)$.
By setting \( \nabla_{\mathbf{b}_{s,k}^*} \mathcal{L} = \mathbf{0} \) and \( \nabla_{\mathbf{p}_{s,k}^*} \mathcal{L} = \mathbf{0} \), respectively, we formulate the closed-form iterative expression for \( \mathbf{b}_{s,k}^{\star} \) as \eqref{beamforming_cf} and \( \mathbf{p}_{s,k}^{\star} \) as \eqref{precoding_cf}.
\begin{figure*}[htbp]
	% ensure that we have normalsize text
	\normalsize
	% Store the current equation number.
	%	\setcounter{MYtempeqncnt}{\value{equation}}
	% Set the equation number to one less than the one
	% desired for the first equation here.
	% The value here will have to changed if equations
	% are added or removed prior to the place these
	% equations are referenced in the main text.
	% \setcounter{equation}{5}
	\vspace{-10mm}
	%\hrulefill
    \begin{align}
		\mathbf{b}_{s, k}^{\star}
         \textstyle =\sqrt{\frac{\kappa_{s,k}\beta_{s,k}}{\kappa_{s,k}+1}}w_{s,k}u_{s, k}\Bigg(w_{s,k}|u_{s,k}|^2\sum_{t\in \mathcal{S}}\sum_{m\in \mathcal{K}}|\mathbf{p}_{t,m}^H\mathbf{g}_{t,k}|^2\mathbf{R}_{t,k}^{\rm ut}+w_{s,k}|u_{s,k}|^2\sigma^2_{k}\mathbf{I}_{N}\Bigg)^{-1}\mathbf{d}_{s, k,0}\mathbf{g}_{s,k}^H\mathbf{p}_{s, k},
        \label{beamforming_cf}\\
	\mathbf{p}_{s, k}^{\star}
\textstyle=\sqrt{\frac{\kappa_{s,k}\beta_{s,k}}{\kappa_{s,k}+1}}w_{s, k} u_{s, k}^*
\Bigg(\sum_{t \in \mathcal{S}}
    \sum_{m \in \mathcal{K}}
    \textcolor{black}{w_{t,m}|u_{t,m}|^{2}}
\mathbf{b}_{t,m}^H\mathbf{R}_{s,m}^{\rm ut}\mathbf{b}_{t,m}\mathbf{g}_{s,m}\mathbf{g}_{s,m}^H
    + \lambda_{s} \mathbf{I}_{M}\Bigg)^{-1}
\mathbf{g}_{s,k}\mathbf{d}_{s,k,0}^H\mathbf{b}_{s, k}.
\label{precoding_cf}
    \end{align}
\hrulefill
\vspace{-6mm}
\end{figure*}
By substituting this expression into problem \eqref{WMMSE1}, it can be reduced to the following optimization problem in terms of Lagrange multipliers \(\boldsymbol{\lambda}=[\lambda_1,\ldots,\lambda_{S}]\):
\begin{align}
	\begin{split}
	      \min_{\boldsymbol{\lambda}}&\quad \epsilon=\sum_{s \in \mathcal S} \sum_{k \in \mathcal{K}} w_{s,k}e_{s,k}(\boldsymbol{\lambda}), \\
		\text{ s.t. } 
        & \quad \sum_{k \in \mathcal{K}} \|\mathbf{p}_{s,k}(\lambda_s)\|_2^2 \leq P_s^{\rm sat},  \forall s \in  \mathcal S,
	\end{split}	\label{WMMSE2}
\end{align}
which transforms solving such a problem into finding the optimal \(\{\lambda_s^{\star}\}_{\forall s}\) with conventional optimization toolboxes. By substituting the obtained \(\{\lambda_s^{\star}\}_{\forall s}\)  into \eqref{precoding_cf}, the optimal precoding vectors and receive vectors  \(\{\mathbf{p}_{s,k}^{\star},\mathbf{b}_{s,k}^{\star}\}_{\forall s, k}\) can be derived without any performance loss.
In conclusion, we optimize the Lagrange multipliers \(\{\lambda_s^{\star}\}_{\forall s}\) in each iteration, and update \(\{{u}_{s,k}^{\star}\}_{\forall s, k}\), \(\{{w}_{s,k}^{\star}\}_{\forall s, k}\), \(\{\mathbf{p}_{s,k}^{\star}\}_{\forall s, k}\) and \(\{\mathbf{b}_{s,k}^{\star}\}_{\forall s, k}\) alternately until convergence, summarized in \algref{opt-wmmse-central}.
\vspace{-2mm}\section{Learning-Based Centralized Multi-Satellite WMMSE Transmission Design}\label{centralized_network}
In multi-satellite massive MIMO systems, large antenna arrays and a high user load make the iterative multi-satellite scheme computationally expensive.  To reduce complexity, we propose a learning-based centralized WMMSE that uses a lightweight neural network to directly infer the key variables of the precoding vectors from \textcolor{black}{slowly varying sCSI: angles, channel power, Rician factor, satellite power budget, and noise power,} thereby avoiding iterative updates.  Exploiting the inherent mapping structure of the WMMSE formulation, the network can be trained in a specific scenario and then applied to dynamic settings with varying numbers of satellites and UTs, enabling scalable deployment while maintaining excellent performance.
\vspace{-3mm}\subsection{Model Knowledge Discovery}
It can be observed that the core of \algref{opt-wmmse-central} lies in obtaining the optimal parameters $\{u_{s,k}^{\star}\}_{\forall s,k}$, $\{w_{s,k}^{\star}\}_{\forall s,k}$, and $\{\lambda_s^{\star}\}_{\forall s}$ through iterative computation, which are then employed to derive the corresponding precoding vectors and receive vectors.
To simplify the prediction, we reformulate the closed-form expression:
\begin{align}
\!\!\!\!\!\mathbf{p}_{s, k}\!\!=\!\! w_{s, k} u_{s, k}^*\!
\Big(\!
\sum_{m \in \mathcal{K}}
\mathbf{g}_{s,m}\mathbf{g}_{s,m}^H
\textcolor{black}{{\varrho}_{s,m}}\!\!+\!\!\lambda_s\mathbf{I}_{M}
\!\!\Big)^{-1}\!\!\!\!\!\!
\mathbf{g}_{s, k}\bar{\mathbf{d}}_{s, k}^H\mathbf{b}_{s, k}.
\label{re_precoding}
\end{align}
Herein, we define the low-dimensional variables ${\varrho}_{s,m}\triangleq \sum_{t \in  \mathcal{S}}w_{t,m}|\mu_{t,m}|^2\mathbf{b}_{t,m}^H \mathbf{R}_{s,m}^{\rm ut}\mathbf{b}_{t,m}, \forall s,m$, and denote \(\bar{\mathbf{d}}_{s,k} \triangleq \mathbb{E}\{\mathbf{d}_{s,k}\}=\sqrt{\frac{\kappa_{s,k}\beta_{s,k}}{\kappa_{s,k}+1}},\mathbf{d}_{s,k,0}\). 
From the closed-form expression in \eqref{re_precoding}, we observe that \(\{w_{s,k},u_{s,k}\}_{\forall s,k}\), \(\{\lambda_s\}_{\forall s}\), \(\{\varrho_{s,m}\}_{\forall s,m}\), and \(\{\mathbf{b}_{s,k}\}_{\forall s,k}\) are all low-dimensional variables. Rather than directly training a neural network to approximate the mapping to the high-dimensional precoding vectors, we instead learn the mapping to these low-dimensional variables, which significantly reduces the learning complexity and relaxes the representational requirements of the neural network.
\vspace{-7mm}
\subsection{Centralized Multi-Satellite Transmission Mapping}
According to \eqref{Rut}, \eqref{mu}, \eqref{weight}, \eqref{WMMSE2}, and \eqref{re_precoding}, we define the input set, which serves as the foundation for the subsequent learning process. In the centralized satellite network architecture, the central satellite combines the local features with the transmitted features from all cooperative satellites, forming the following centralized network input set $\mathcal{H}^{\rm cen} \triangleq\left\{\mathcal{H}_{s}\right\}_{\forall s \in \mathcal{S}}$:
\vspace{-2mm}\begin{align}
\!\!\!\!\!\!\mathcal{H}^{\rm cen}
\!\!\triangleq\!\!\left\{ \!\left\{\!\Phi_{s,k},\!\Theta_{s,k},\!\beta_{s,k},\! \kappa_{s,k}, \!\boldsymbol{\Sigma}_{s,k},\!\sigma_{k}^2\!\right\}_{\!\forall k \in \mathcal{K}} \!\cup\!\!  \left\{\!P_s^{\rm sat}\!\right\}\!\right\}_{\!\forall s \in \mathcal S}\!,
\end{align}
herein, we simplify the angular information as $\Phi_{s,k}=\{{\phi}^{\rm sat}_{s,k},{\phi}^{\rm ut}_{s,k,0}\}$ and $\Theta_{s,k}=\{{\theta}^{\rm sat}_{s,k},{\theta}^{\rm ut}_{s,k,0}\}$, both of which can be derived from the satellite position, satellite attitude, and UT GNSS information. The central satellite feature set $\mathcal{H}^{\rm cen}$ comprises the angular information of the satellite and UTs, the average channel power $\beta_{s,k}$, the Rician factor $\kappa_{s,k}$, the spatial covariance matrix $\boldsymbol{\Sigma}_{s,k}$ of the NLoS component, and the receiver noise power $\sigma_{k}^2$.
Based on \(\mathcal{H}^{\rm cen}\), we construct the centralized network input tensors, namely the third-order tensors \(\tA \in \mathbb{R}^{S \times K \times 6}\), \(\tR \in \mathbb{R}^{S \times K \times 2(M^2+N^2)}\), and \(\tB \in \mathbb{R}^{S \times K \times 2}\):
\vspace{-5mm}\begin{subequations}
\begin{align}
\tA_{[s,k,:]} &\!=\! \big[ \phi^{\rm sat}_{s,k},\, \theta^{\rm sat}_{s,k},\, \phi^{\rm ut}_{s,k,0},\, \theta^{\rm ut}_{s,k,0}, \sigma^2_{s,k}, P_{s,k}^{\rm sat} \big],\\
\tR_{[s,k,:]} &\!=\! \big[ \operatorname{vec}(\Re(\mathbf{R}_{s,k}^{\rm ut})),\,
\operatorname{vec}(\Im(\mathbf{R}_{s,k}^{\rm ut})),\nonumber\\
&\hspace{1.1em}\operatorname{vec}(\Re(\mathbf{R}_{s,k}^{\rm sat})),\,
\operatorname{vec}(\Im(\mathbf{R}_{s,k}^{\rm sat})) \big],\\
\tB_{[s,k,:]} &\!=\! \big[\beta_{s,k},\, \kappa_{s,k}\big],
\end{align}
\end{subequations}
where $\sigma_{s,k}^2 \triangleq \sigma^2_k, \forall s \in \mathcal{S}$ and $P_{s,k}^{\rm sat} \triangleq P_{s}^{\rm sat}, \forall k \in \mathcal{K}$. We collect the low-dimensional variables from \algref{opt-wmmse-central} into the following predicted tuple:
\begin{align}
\mathsf{O}^{\rm cen}&\triangleq 
\Big(
\{w_{s,k},u_{s,k}\}_{\forall s,k},\,
\{\lambda_s\}_{\forall s},\,
\{\varrho_{s,k}\}_{\forall s,k},\,
\{\mathbf{b}_{s,k}\}_{\forall s,k}
\Big)\nonumber\\
&=\Big(\mathbf{w},\boldsymbol{\mu},\boldsymbol{\lambda},\boldsymbol{\varrho},\tb\Big).
\end{align}
\(\boldsymbol{\lambda}\in \mathbb{R}^{S}\), \(\mathbf{w}\in \mathbb{R}^{S \times K}\) and \(\boldsymbol{\mu}, \boldsymbol{\varrho}\in \mathbb{C}^{S \times K}\), and \(\tb\in \mathbb{R}^{S \times K \times N}\).  Since problem~\eqref{WMMSE2} is non-convex, \algref{opt-wmmse-central} can be viewed as a mapping from the CSI tensors
\(\big(\tA,\,\tR,\,\tB\big)\) to an optimal output tuple \(\mathsf{O}^{\rm cen}\), i.e.,
\begin{align}
G^{\rm cen}(\tA,\,\tR,\,\tB)=\mathsf{O}^{\rm cen}.
\end{align}
In the \textcolor{black}{offline} learning stage, we aim to design \(G^{\rm cen}\) with a consistent structure that preserves the inherent equivariance of the problem.
We denote $\pi_S\in\mathbb{S}_S$ and $\pi_K\in\mathbb{S}_K$ as arbitrary permutations of satellite and UT indices, respectively. Let $\circ_n$ denote the permutation action on the $n$-th dimension. 
For $\tX\in\{\tA,\tR,\tB\}$, we write $\pi_S\circ_1\tX$ and $\pi_K\circ_2\tX$ for permuting the satellite dimension and UT dimension, respectively.
\begin{ppn}\label{ppn_cen}
For any $\pi_S\in\mathbb{S}_S$ and $\pi_K\in\mathbb{S}_K$, the achieved objective value of problem \eqref{WMMSE2} is equivariance under consistent permutations, i.e.,
\begin{align}
R\langle&\mathsf{O}^{\rm cen};\tA,\tR,\tB\rangle
\nonumber\\
=R\langle&\pi_S\circ_1\mathbf{w},\pi_S\circ_1\boldsymbol{\mu},\pi_S\circ_1\boldsymbol{\lambda},\pi_S\circ_1\boldsymbol{\varrho},\pi_S\circ_1\tb;\nonumber\\
&\pi_S\circ_1\tA,\ \pi_S\circ_1\tR,\ \pi_S\circ_1\tB \rangle, \label{cen_perm_obj_S}
\end{align}
\begin{align}
R\langle&\mathsf{O}^{\rm cen};\tA,\tR,\tB\rangle
\nonumber\\
=R\langle &\pi_K\circ_2\mathbf{w},\pi_K\circ_2\boldsymbol{\mu},\boldsymbol{\lambda},\pi_K\circ_2\boldsymbol{\varrho},\pi_K\circ_2\tb;\nonumber\\
&\pi_K\circ_2\tA,\ \pi_K\circ_2\tR,\ \pi_K\circ_2\tB \rangle. \label{cen_perm_obj_K}
\end{align}
\end{ppn}
\begin{pf} The proof is omitted for clarity \cite{wang2024,wang2026}.
\end{pf}
Based on \eqref{cen_perm_obj_S} and \eqref{cen_perm_obj_K}, the mapping $G^{\rm cen}$ satisfies the following tensor equivariance (TE) properties:
\begin{align}
\!\!\!G^{\rm cen}(\pi_S\!\circ_1\!\tA,\pi_S\!\circ_1\!\tR,\pi_S\!\circ_1\!\tB)
\!=\!\pi_S\circ_1 G^{\rm cen}(\tA,\tR,\tB), \label{TE_sat_cen}
\end{align}
\begin{align}
\!\!\!G^{\rm cen}(\pi_K\!\circ_2\!\tA,\pi_K\!\circ_2\!\tR,\pi_K\!\circ_2\!\tB)
\!=\!\pi_K\circ_2 G^{\rm cen}(\tA,\tR,\tB),\label{TE_user_cen}
\end{align}
which indicates that $G^{\rm cen}$ is permutation-equivariant with respect to both satellite and UT index reorderings. Moreover, although \eqref{WMMSE2} is non-convex and multiple optimal tuples may exist, \eqref{TE_sat_cen} and \eqref{TE_user_cen} still hold since the set of optimal solutions is closed under the above permutation actions.
\begin{figure*}[tbp]
\centering
	% ensure that we have normalsize text
	\normalsize
	% Store the current equation number.
	%	\setcounter{MYtempeqncnt}{\value{equation}}
	% Set the equation number to one less than the one
	% desired for the first equation here.
	% The value here will have to changed if equations
	% are added or removed prior to the place these
	% equations are referenced in the main text.
	% \setcounter{equation}{5}
	\vspace{-9mm}
	%\hrulefill
\includegraphics[width=0.73 \linewidth,trim=0.1cm 0.1cm 0.1cm 0.1cm,clip]{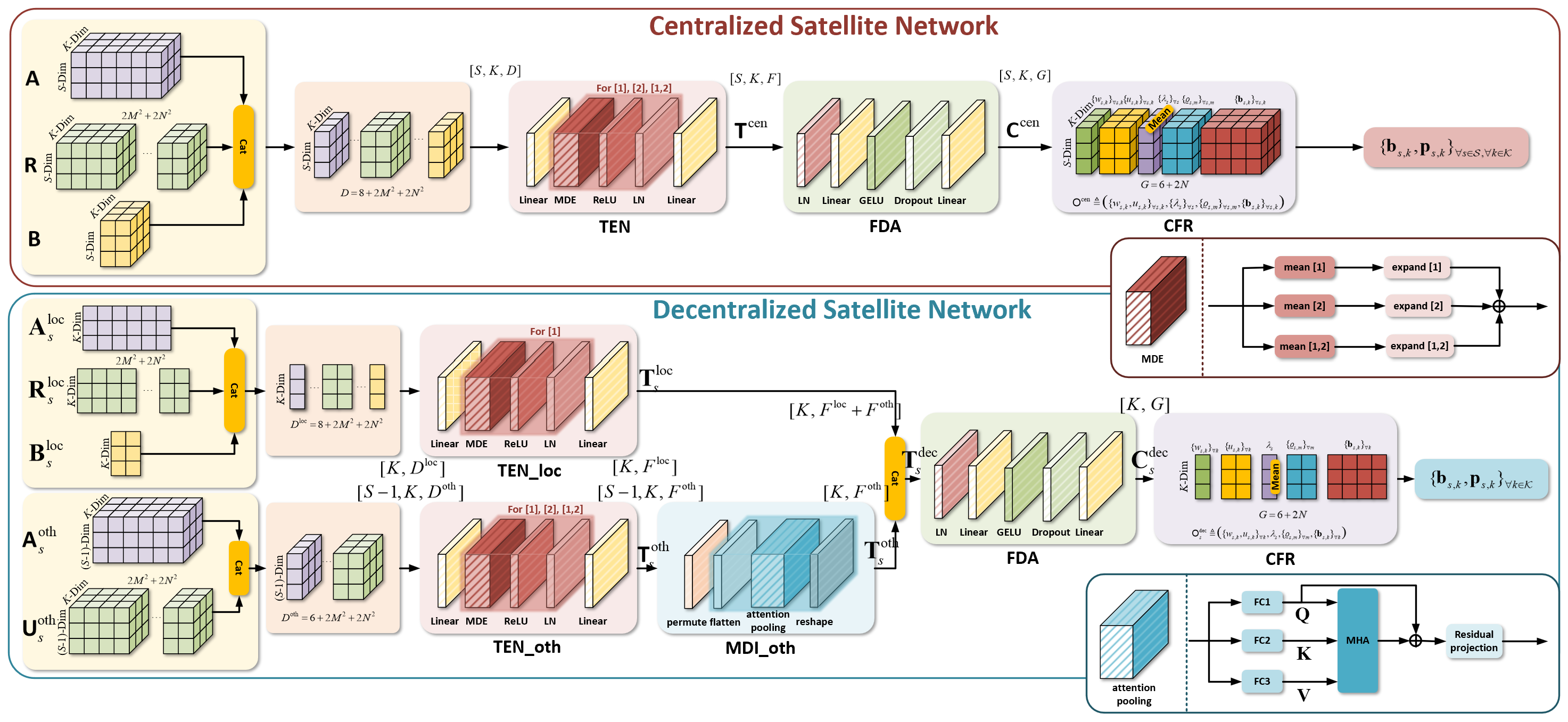}
 \vspace{-2mm}	\caption{Satellite network architectures: centralized and decentralized.}
\label{central_decentralized_network}\vspace{-5mm}
\end{figure*}
\subsection{Centralized Multi-Satellite Network Design}
As specifically shown in \figref{central_decentralized_network}, the proposed centralized network adopts a three-stage architecture, namely the tensor-equivariant neural (TEN) network, the centralized feature dimension adapter (FDA), and the centralized closed-form recovery (CFR). Specifically, TEN extracts permutation-equivariant latent features from the input tensors, FDA projects these features to the low-dimensional variables to be predicted, and CFR substitutes the predicted variables into the analytical expressions to recover the precoding and receive vectors.
\subsubsection{TEN Network}
In the preprocessing stage, we employ a TEN network to transform the concatenated input tensor into a latent representation while preserving permutation equivariance with respect to both the satellite and UT dimensions:
\begin{align}
\tT^{\rm cen}=N_{\rm TEN}^{\rm cen}\!\left(\tA,\tR,\tB\right) \in\mathbb{R}^{S\times K\times F}.
\end{align}
We first define a multi-dimensional linear layer that applies an affine mapping along the last dimension:
$
\mathrm{Linear}(\tX;\mathbf W,\mathbf b)
\triangleq 
\tX \times \mathbf W + \mathbf{1}_{M_1,\ldots,M_N,1}\odot \mathbf b^{T}$,
where $\tX\in\mathbb{R}^{M_1\times\cdots\times M_N\times D_{\rm I}}$, $\mathbf W\in\mathbb{R}^{D_{\rm I}\times D_{\rm O}}$, and $\mathbf b\in\mathbb{R}^{D_{\rm O}}$.
We construct the centralized input tensor $\tX^{\rm cen}\in\mathbb{R}^{S\times K\times D}$, where $D=8+2M^2+2N^2$.  The TEN consists of an input linear layer followed by \(L\) stacked blocks, each comprising a multi-dimensional equivariant (MDE) module, a rectified linear unit (ReLU) \cite{Nair2010_ReLU}, a layer normalization (LN), and an output linear layer:
\begin{subequations}
\begin{align}
\tX^{\rm cen}&=[\tA,\tR,\tB]_3\in\mathbb{R}^{S\times K\times D},\\
\tH^{(0)}_{\rm TEN}
&=\mathrm{Linear}\!\left(\tX^{\rm cen};\mathbf W_{1}^{\rm cen},\mathbf b_{1}^{\rm cen}\right)\in\mathbb{R}^{S\times K\times d_{\rm h}},\\
\tH^{(\ell)}_{\rm TEN}
&=\!\operatorname{LN}\!\Big(\!\operatorname{ReLU}\!\big(\operatorname{MDE}_{\ell}(\tH^{(\ell-1)}_{\rm TEN})\big)\!\Big),\ell\!=\!1,\ldots,L,\\
\tT^{\rm cen}
&=\mathrm{Linear}\!\left(\tH^{(L)}_{\rm TEN};\mathbf W_{2}^{\rm cen},\mathbf b_{2}^{\rm cen}\right),
\end{align}
\end{subequations}
where $d_{\rm h}$ denotes the hidden dimension, \(\mathbf{W}_1^{\rm cen}\in\mathbb{R}^{D\times d_{\rm h}}\), \(\mathbf{W}_2^{\rm cen}\in\mathbb{R}^{d_{\rm h}\times F}\), and \(\mathbf{b}_1^{\rm cen}\in\mathbb{R}^{d_{\rm h}},\mathbf{b}_{2}^{\rm cen}\in\mathbb{R}^{F}\).
In each layer, \(\operatorname{MDE}_{\ell}(\cdot)\) is implemented as a linear combination of mean-and-repeat patterns defined on subsets of the two equivariant dimensions \(\{S, K\}\) \cite{wang2024}. Consequently, the overall mapping \(N_{\rm TEN}^{\rm cen}\) is permutation-equivariant: any permutation applied to the satellite and UT indices of \(\tX^{\rm cen}\) induces the same permutation on \(\tT^{\rm cen}\), which is consistent with the intrinsic symmetry of the considered multi-satellite multi-user system.

\subsubsection{Centralized FDA Module}
After obtaining \(\tT^{\rm cen}\in\mathbb{R}^{S\times K\times F}\) from the TEN network, we apply a lightweight FDA to project the features to \(\tC^{\rm cen}\in\mathbb{R}^{S\times K\times G}\) with \(G=6+2N\):
\begin{align}
\tC^{\rm cen} 
= N_{\rm FDA}^{\rm cen}\!\left(\tT^{\rm cen}\right)\in\mathbb{R}^{S\times K\times G}.
\end{align}
Specifically, the centralized FDA is implemented as one LN and two linear layers with a Gaussian error linear unit (GELU) activation \cite{Hendrycks2016_GELU} and a dropout layer:
\begin{subequations}
\begin{align}
\hat{\tT}^{\rm cen}
&= \operatorname{LN}\!\left(\tT^{\rm cen}\right),\\
\tH_{\rm FDA}^{\rm cen} 
&= \operatorname{GELU}\!\Big(\mathrm{Linear}\!\left(\hat{\tT}^{\rm cen};\mathbf{W}_{3}^{\rm cen},\mathbf{b}_{3}^{\rm cen}\right)\Big),\\
\tilde{\tH}_{\rm FDA}^{\rm cen} 
&= \operatorname{Dropout}\!\left(\tH_{\rm FDA}^{\rm cen}\right),\\
\tC^{\rm cen} 
&= \mathrm{Linear}\!\left(\tilde{\tH}_{\rm FDA}^{\rm cen};\mathbf{W}_{4}^{\rm cen},\mathbf{b}_{4}^{\rm cen}\right),
\end{align}
\end{subequations}
where \(\mathbf{W}_3^{\rm cen}\in\mathbb{R}^{F\times G}\), \(\mathbf{W}_4^{\rm cen}\in\mathbb{R}^{G\times G}\), and \(\mathbf{b}_3^{\rm cen},\mathbf{b}_4^{\rm cen}\in\mathbb{R}^{G}\).
\subsubsection{Centralized CFR Module}
Given the tensor \(\tC^{\rm cen}\) produced by the FDA module, the predicted tuple \(\mathsf{O}^{\rm cen}\) is obtained via the closed-form recovery mapping, i.e.,
\begin{align}
\mathsf{O}^{\rm cen}=N_{\rm CFR}^{\rm cen}\big(\tC^{\rm cen}\big).
\end{align}
The recovered variables are then substituted into the analytical expressions to compute the centralized precoding vectors $\{\mathbf{p}_{s,k}\}_{\forall s,k}$ in \eqref{re_precoding}. Notably, the receive vectors $\{\mathbf{b}_{s,k}\}_{\forall s,k}$ are predicted directly, since they are low-dimensional.

The whole centralized network design establishes an end-to-end prediction, ensuring both physical interpretability and learning flexibility. The network components are collectively referred to as TEN-FDA-CFR (TFC) network, and the overall process is summarized in \algref{net-wmmse-central}.
For the centralized network, the TEN network comprises \(L\) MDE layers, each incorporating \(\rho_{2}\) patterns, resulting in a complexity of \(\mathcal{C}_{\rm TEN}=\mathcal{O}(SK(D d_{\rm h}+L\rho_{2} d_{\rm h}^{2}+d_{\rm h}F))\). The FDA incurs \(\mathcal{C}_{\rm FDA}=\mathcal{O}(SK(FG+G^{2}))\). Finally, the CFR involves \(\mathcal{C}_{\rm CFR}=\mathcal{O}(SKM^{3})\), which typically dominates the overall inference complexity.

\begin{algorithm}[t]
  \caption{Centralized TFC-Net-Based WMMSE}
  \label{net-wmmse-central}
  \begin{algorithmic}[1]
    \STATE \textbf{Input:}$\{\Phi_{s,k},\!\Theta_{s,k},\!\beta_{s,k},\!\kappa_{s,k},\!\boldsymbol{\Sigma}_{s,k}\}_{\forall s,k}$,$\{\sigma_{k}^2\}_{\forall k}$,$\{P_s^{\rm sat}\}_{\forall s}$.
    \STATE Construct the centralized input tensors $\tA,\tR,\tB$.
    \STATE \textbf{Training:}
    \STATE Initialize $\mathcal{N}_{\rm cen}^{(0)}=\{N_{\rm TEN}^{\rm cen},N_{\rm FDA}^{\rm cen},N_{\rm CFR}^{\rm cen}\}^{(0)}$, $n=0$.
    \STATE \textbf{for} $n < N_{\rm ep}$ \textbf{do}
      \STATE \quad \textbf{forward:}
      \STATE \quad\quad Apply the TEN network: $\tT^{\rm cen} = N_{\rm TEN}^{\rm cen}\big(\tA,\tR,\tB\big)$.
      \STATE \quad\quad Apply the FDA module: $\tC^{\rm cen} = N_{\rm FDA}^{\rm cen}\big(\tT^{\rm cen}\big)$.
      \STATE \quad\quad Recover variables: $\mathsf{O}^{\rm cen} = N_{\rm CFR}^{\rm cen}\big(\tC^{\rm cen}\big)$.
      \STATE \quad\quad Obtain $\{\mathbf{p}_{s,k},\!\mathbf{b}_{s,k}\}_{\forall s,k}$ from $\mathsf{O}^{\rm cen}$ with \eqref{re_precoding}.
      \STATE \quad \textbf{backward:}
      \STATE \quad\quad Update $\mathcal{N}_{\rm cen}^{(n)}$ via gradient descent on $\mathcal{L}_{\rm wsr}$. 
      \STATE \quad \quad Update $n=n+1$.
    \STATE \textbf{end for}
   \STATE Store the learned variables of $\mathcal{N}_{\rm cen}$.
    \STATE \textbf{Inference:}  Use $\mathcal{N}_{\rm cen}$ to obtain $\{\mathbf{p}_{s,k},\mathbf{b}_{s,k}\}_{\forall s,k}$.
  \end{algorithmic}
\end{algorithm}
\section{Decentralized Multi-Satellite Transmission Architecture}\label{decentralized_arch}
In multi-satellite cooperative transmission scenarios, conventional coordination typically relies on a central satellite acting as the controller. In this architecture, each satellite uploads its locally measured CSI to the central satellite, which jointly computes the precoders and then broadcasts \(\{\mathbf{p}_{s,k}\}_{\forall s,k}\) back to the cooperating satellites for execution, as in \secref{centralized_optimization} and \secref{centralized_network}. However, centralized coordination can incur substantial ISL overhead and processing latency, which is particularly problematic in LEO systems where ISL resources are limited. In large LEO constellations, satellites located in counter-rotating orbital seams or across different orbital planes may experience large relative motion, making ISL establishment difficult and leading to link instability \cite{Chu2025,Beatriz2019}. To address this issue, we develop a decentralized transmission architecture that removes the reliance on a central satellite, as illustrated in \figref{central_decentralized_arch}. Specifically, satellites periodically exchange only SSI, including satellite position and attitude, which is low-rate and delay-tolerant. With SSI and UT GNSS, each satellite can locally anticipate the transmission behavior of neighboring satellites, thereby designing its own precoder without exchanging high-dimensional CSI. 
\vspace{-2mm}\subsection{Feasibility of Decentralized Multi-Satellite Architecture}
The feasibility of decentralized multi-satellite precoding can be revealed from the closed-form expression of the WMMSE precoder in \eqref{re_precoding}.
\begin{remark} The cross-satellite coupling in \eqref{re_precoding} is mainly embodied in a set of low-dimensional variables, including the aggregated interference terms $\{\varrho_{s,m}\}_{\forall m}$, the scalar coefficients $\{w_{s,k},u_{s,k}\}$, the per-satellite Lagrange multiplier $\lambda_s$, and the receive vector $\mathbf{b}_{s,k}$.
According to \eqref{mu}, \eqref{weight}, \eqref{WMMSE2}, and \eqref{beamforming_cf}, these variables are determined by the satellite steering vectors $\mathbf{g}_{s,k}$ and $\mathbf{d}_{s,k,0}$, the channel correlation matrices $\mathbf{R}_{s,k}^{\rm sat}$ and $\mathbf{R}_{s,k}^{\rm ut}$, the channel power $\beta_{s,k}$, the Rician factor $\kappa_{s,k}$, the noise power $\sigma_k^2$, and the satellite power budget $P_s^{\rm sat}$.\end{remark}
Importantly, the high-dimensional steering vectors can be constructed from low-dimensional angular parameters, which can in turn be estimated from the positions of satellites and UTs. Moreover, the high-dimensional correlation matrices of other satellites can be further approximated from these steering vectors. Therefore, instead of relying on high-dimensional CSI, the cross-satellite coupling variables can be predicted from local  satellite CSI, together with the angular information and approximate correlation matrices of other satellites, as well as the noise power and satellite power budgets. 

From an ISL communication perspective, centralized architectures incur substantial overhead due to high-dimensional CSI exchange and frequent precoder dissemination. In contrast, the above dependency analysis suggests that lightweight SSI, such as satellite position, attitude, and power budgets, may already provide sufficient information for each satellite to anticipate the influence of neighboring transmissions and enable a feasible local precoding design with lightweight inter-satellite interaction. Motivated by this observation, we propose a decentralized architecture in which each satellite performs local inference based on its local CSI and the SSI of other satellites. This significantly reduces inter-satellite overhead while preserving the essential coupling among satellites through the prediction of low-dimensional variables.
\begin{figure}[!t]
\centering
		\vspace{-9mm}
	%\hrulefill
\includegraphics[width=0.72\linewidth,trim=0.1cm 0.1cm 0.1cm 0.1cm,clip]{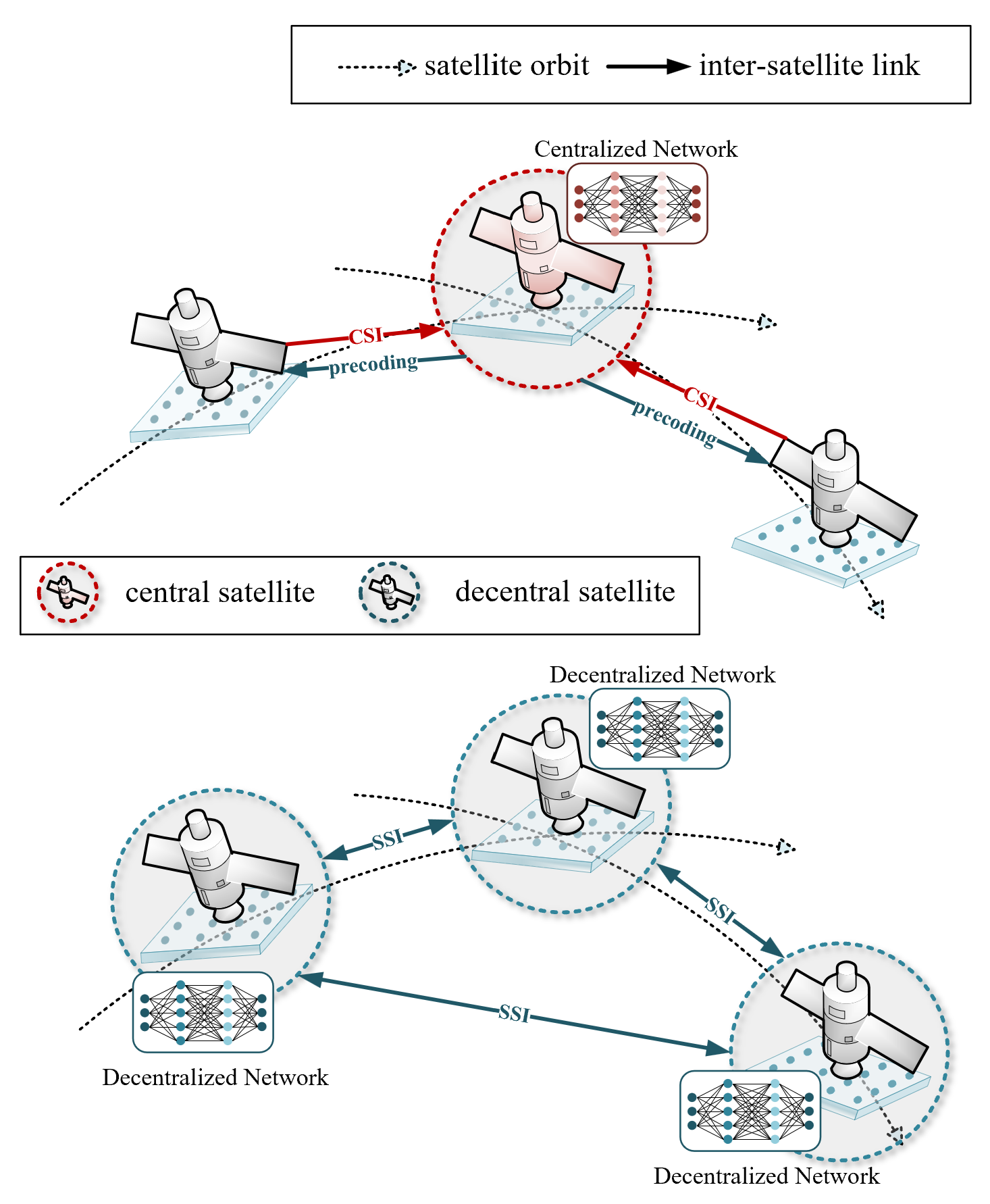}
 \vspace{-2mm}	\caption{Multi-satellite cooperative architectures: centralized and decentralized.}
\label{central_decentralized_arch}	\vspace{-6mm}
\end{figure}
\vspace{-1mm}
\subsection{Decentralized Multi-Satellite Transmission Mapping}
Assuming that only SSI is periodically exchanged among decentral satellites, the decentralized feature set is reconstructed with two components: the local-satellite features $\mathcal{H}_{s}^{\rm loc}$ and the other-satellite features $\mathcal{H}_{s}^{\rm oth}$:
\begin{align}
\!\!\!\mathcal{H}_{s}^{\rm loc} \!\triangleq \!
\left\{\Phi_{s,k},\Theta_{s,k},\beta_{s,k}, \kappa_{s,k}, \boldsymbol{\Sigma}_{s,k},\sigma_{k}^2\right\}_{\forall k \in \mathcal{K}}\! \cup \! \left\{P_s^{\rm sat}\right\},
\end{align}
\begin{align}
\mathcal{H}_{s}^{\rm oth} \!\triangleq \!
\{\Phi_{t,m},\Theta_{t,m}\}_{\forall t \in \mathcal S \setminus s, \forall m \in \mathcal{K}}
\!\cup\! \{\sigma_{m}^2\}_{\forall m \in \mathcal K}
\!\cup \! \{P_t^{\rm sat}\}_{\forall t \in \mathcal S \setminus s}.
\end{align}
The decentralized input set for the \(s\)-th satellite is:
\begin{align}
\mathcal{H}_{s}^{\rm dec} \triangleq\left\{\mathcal{H}_{s}^{\rm loc},\mathcal{H}_{s}^{\rm oth}\right\},  \forall s \in \mathcal{S}.
\end{align}
Based on $\mathcal{H}_{s}^{\rm loc}$, we construct the local-satellite input second-order tensors
\(\mathbf{A}_s^{\rm loc} \in \mathbb{R}^{K \times 6}\),
\(\mathbf{R}_s^{\rm loc} \in \mathbb{R}^{K \times 2(M^2+N^2)}\),
and \(\mathbf{B}_s^{\rm loc} \in \mathbb{R}^{K \times 2}\) as follows:
\begin{subequations}
\begin{align}
&{\mathbf{A}_{s}^{\rm loc}}_{[m,:]}
=\left[{\phi}^{\rm sat}_{s,m},{\theta}^{\rm sat}_{s,m},{\phi}^{\rm ut}_{s,m,0},{\theta}^{\rm ut}_{s,m,0},\sigma^2_{s,m},P_{s,m}^{\rm sat}\right],\\
&{\mathbf{R}_{s}^{\rm loc}}_{[m,:]}
=\big[ \operatorname{vec}(\Re(\mathbf{R}_{s,m}^{\rm ut})),\,
\operatorname{vec}(\Im(\mathbf{R}_{s,m}^{\rm ut})),\nonumber\\
&\hspace{5em}\operatorname{vec}(\Re(\mathbf{R}_{s,m}^{\rm sat})),\,
\operatorname{vec}(\Im(\mathbf{R}_{s,m}^{\rm sat})) \big],\\
&{\mathbf{B}_{s}^{\rm loc}}_{[m,:]}
=\left[\beta_{s,m},\,\kappa_{s,m}\right],
\end{align}
\end{subequations}
where $\sigma^2_{s,m} \triangleq \sigma_m^2$ and $P_{s,m}^{\rm sat} \triangleq P_{s}^{\rm sat}, \forall m \in \mathcal{K}$.
Based on $\mathcal{H}_{s}^{\rm oth}$, we construct the other-satellite input third-order tensors
\(\tA^{\rm oth}_s \in \mathbb{R}^{(S-1) \times K \times 6}\) and
\(\tU^{\rm oth}_s \in \mathbb{R}^{(S-1) \times K \times 2(M^2+N^2)}\) as:
\begin{subequations}
\begin{align}
&{\tA_s^{\rm oth}}_{[t,m,:]}
\!\!=\!\!\left[{\phi}^{\rm sat}_{t,m},{\theta}^{\rm sat}_{t,m},{\phi}^{\rm ut}_{t,m,0},{\theta}^{\rm ut}_{t,m,0},\sigma^2_{t,m}, P_{t,m}^{\rm sat}\right], \ \forall t \in \mathcal{S} \setminus s,\\
&{\tU_s^{\rm oth}}_{[t,m,:]}
\!\!=\!\!\big[\operatorname{vec}\big(\Re\big(\mathbf{d}_{t,m,0}\mathbf{d}_{t,m,0}^H\big)\big),
\operatorname{vec}\big(\Im\big(\mathbf{d}_{t,m,0}\mathbf{d}_{t,m,0}^H\big)\big),\nonumber\\
&\hspace{1.1em}
\quad\operatorname{vec}\big(\Re\big(\mathbf{g}_{t,m}\mathbf{g}_{t,m}^H\big)\big),
\operatorname{vec}\big(\Im\big(\mathbf{g}_{t,m}\mathbf{g}_{t,m}^H\big)\big)\big], \ \forall t \in \mathcal{S} \setminus s,
\end{align}
\end{subequations}
where $\sigma^2_{t,m} \triangleq \sigma_m^2, \forall t \in \mathcal{S} \setminus s$ and $P_{t,m}^{\rm sat} \triangleq P_{t}^{\rm sat}, \forall m \in \mathcal{K}$.
For the $s$-th satellite, we collect the low-dimensional variables associated with its local transmission into the predicted tuple:
\begin{align}
\mathsf{O}_{s}^{\rm dec}&\triangleq 
\Big(
\{w_{s,k},u_{s,k}\}_{\forall k},\,
\lambda_s,\,
\{\varrho_{s,k}\}_{\forall k},\,
\{\mathbf{b}_{s,k}\}_{\forall k}
\Big)\nonumber\\
&=\Big(\mathbf{w}_s,\boldsymbol{\mu}_s,\lambda_s,\boldsymbol{\varrho}_s,\mathbf{b}_s\Big).
\end{align}
The decentralized inference at the $s$-th satellite can be viewed as a mapping from the local and other-satellite CSI tensors
\(\big(\mathbf{A}_s^{\rm loc},\mathbf{R}_s^{\rm loc},\mathbf{B}_s^{\rm loc},\tA_s^{\rm oth},\tU_s^{\rm oth}\big)\)
to an optimal output tuple:
\begin{align}
G^{\rm dec}(\mathbf{A}_s^{\rm loc},\mathbf{R}_s^{\rm loc},\mathbf{B}_s^{\rm loc},\tA_s^{\rm oth},\tU_s^{\rm oth})
=\mathsf{O}_{s}^{\rm dec}.
\end{align}
Here, we aim to design \(G^{\rm dec}\) to preserve the permutation-induced invariance and equivariance inherent in the problem.
Let $\pi_{S-1}\in\mathbb{S}_{S-1}$ be arbitrary permutations the ``other-satellite'' indices (i.e., the first dimension of $\tA_s^{\rm oth}$ and $\tU_s^{\rm oth}$).  
\begin{ppn}\label{ppn_dec}
For any $\pi_K\in\mathbb{S}_K$ and $\pi_{S-1}\in\mathbb{S}_{S-1}$, the $s$-th satellite inference ($\forall s$) satisfies:
\begin{align}
R\langle & \mathsf{O}_{s}^{\rm dec};\mathbf{A}_s^{\rm loc},\mathbf{R}_s^{\rm loc},\mathbf{B}_s^{\rm loc},\tA_s^{\rm oth},\tU_s^{\rm oth}\rangle \nonumber\\
=R\langle &\mathsf{O}_{s}^{\rm dec}\!;\mathbf{A}_s^{\rm loc}\!\!,\mathbf{R}_s^{\rm loc}\!\!,\mathbf{B}_s^{\rm loc}\!\!,\pi_{S-1}\!\circ_1\tA_s^{\rm oth}\!\!,\pi_{S-1}\!\circ_1\!\tU_s^{\rm oth}\rangle, \label{dec_perm_obj_S}
\end{align}
\begin{align}
&R\langle  \mathsf{O}_{s}^{\rm dec};\mathbf{A}_s^{\rm loc},\mathbf{R}_s^{\rm loc},\mathbf{B}_s^{\rm loc},\tA_s^{\rm oth},\tU_s^{\rm oth}\rangle \nonumber\\
=&R\langle \pi_K\!\circ_1\mathbf{w}_s,\pi_K\!\circ_1\boldsymbol{\mu}_s,\lambda_s,\pi_K\!\circ_1\boldsymbol{\varrho}_s,\pi_K\!\circ_1\mathbf{b}_s;\nonumber\\
&\!\!\!\!\!\!\pi_K\!\circ_1\!\mathbf{A}_s^{\rm loc}\!\!,\pi_K\!\circ_1\!\mathbf{R}_s^{\rm loc}\!\!,\pi_K\!\circ_1\!\mathbf{B}_s^{\rm loc}\!\!,
\pi_K\!\circ_2\!\tA_s^{\rm oth}\!\!,\pi_K\!\circ_2\!\tU_s^{\rm oth}\rangle. 
\label{dec_perm_obj_K}
\end{align}
\end{ppn}
\begin{pf}
The proof is omitted for clarity \cite{wang2024,wang2026}.
\end{pf}
Based on \eqref{dec_perm_obj_S} and \eqref{dec_perm_obj_K}, the mapping $G^{\rm dec}$ satisfies:
\begin{align}
&G^{\rm dec}(\mathbf{A}_s^{\rm loc},\mathbf{R}_s^{\rm loc},\mathbf{B}_s^{\rm loc},\pi_{S-1}\!\circ_1\tA_s^{\rm oth},\pi_{S-1}\!\circ_1\tU_s^{\rm oth})
\nonumber\\
&=G^{\rm dec}(\mathbf{A}_s^{\rm loc},\mathbf{R}_s^{\rm loc},\mathbf{B}_s^{\rm loc},\tA_s^{\rm oth},\tU_s^{\rm oth}), 
\label{TE_dec_S}\\
&\!\!\!G^{\rm dec}(\pi_K\!\circ_1\mathbf{A}_s^{\rm loc}\!,\pi_K\!\circ_1\!\mathbf{R}_s^{\rm loc}\!,\pi_K\!\circ_1\!\mathbf{B}_s^{\rm loc}\!,\pi_K\!\circ_2\!\tA_s^{\rm oth}\!,\pi_K\!\circ_2\!\tU_s^{\rm oth})
\nonumber\\
&=\pi_K\!\circ_1 G^{\rm dec}(\mathbf{A}_s^{\rm loc},\mathbf{R}_s^{\rm loc},\mathbf{B}_s^{\rm loc},\tA_s^{\rm oth},\tU_s^{\rm oth}),
\label{TE_dec_K}
\end{align}
\begin{remark}
Note that the permutation equivariance of \(G^{\rm dec}\) in \eqref{TE_dec_S} is defined with respect to the other \((S-1)\) satellites, rather than the full set of \(S\) satellites considered in the centralized inference setting \eqref{TE_sat_cen}.
\end{remark}
According to \eqref{dec_perm_obj_S}, permuting the order of the other-satellite set does not affect the inference result, and the output does not include an other-satellite dimension.
Denote the pooling-by-multihead-attention (PMA) operator $\mathrm{PMA}(\cdot)$ and the multi-dimensional invariance (MDI) module $\mathrm{MDI}_{\{1\}}(\cdot)$ \cite{wang2024}. 
Considering $\tX$ a tensor with first dimension of $S-1$, we model an invariant mapping $f: \mathbb{R}^{(S-1)\times K \times D_x} \rightarrow \mathbb{R}^{K \times D_x}$, and its expression is given by: $
\mathbf{X}=\mathrm{MDI}_{\{1\}}\!\left(\tX\right)
=[\mathrm{PMA}\circ_{1}\tX] \in \mathbb{R}^{K \times D_x}$. $\circ_{1}$ denotes the computation of PMA along the first dimension, and $[\cdot]$ removes empty dimensions of size 1 generated during this process.
\vspace{-2mm}\section{Learning-Based Decentralized WMMSE Network Design}\label{decentralized_network}
As shown in \figref{central_decentralized_network}, by leveraging the per-satellite WMMSE mapping structure, the proposed decentralized network is deployed locally at each satellite with shared decentralized network parameters, enabling parallel inference across satellites without centralized coordination while naturally generalizing to different scenarios.
\vspace{-2mm}\subsection{Decentralized Multi-Satellite Network Design}
For the \(s\)-th satellite, the network adopts a three-stage pipeline consisting of a dual-branch TEN (DB-TEN) network, a decentralized FDA module, and a decentralized CFR module, as in \figref{central_decentralized_network}. Specifically, the DB-TEN processes the local-satellite and other-satellite features in parallel; the other-satellite branch further incorporates an MDI module to aggregate information along the other-satellite dimension. The resulting fused features are then mapped by the decentralized FDA to the low-dimensional variables, which are finally recovered by the decentralized CFR.
\subsubsection{DB-TEN Network}
To accommodate the heterogeneous feature structures in the decentralized architecture, we employ two TEN branches with different tensor orders to process the local-satellite and other-satellite inputs, respectively.

We first apply a 1D-equivariant TEN (with respect to the UT dimension) to the local-satellite input tensors:
\begin{align}
\mathbf{T}^{\rm loc}_s=N_{\rm TEN}^{\rm loc}\big(\mathbf{A}^{\rm loc}_s,\mathbf{R}^{\rm loc}_s,\mathbf{B}^{\rm loc}_s\big) \in \mathbb{R}^{K \times F^{\rm loc}},
\end{align}
The local TEN adopts an architecture with an initial linear layer, followed by \(L\) stacked MDE-ReLU-LN blocks, and a final linear layer, given by: 
\begin{subequations}
\begin{align}
\mathbf{X}_s^{\rm loc}&=[\mathbf{A}_s^{\rm loc},\mathbf{R}_s^{\rm loc},\mathbf{B}_s^{\rm loc}]_2\in\mathbb{R}^{K\times D^{\rm loc}},\\
\mathbf{H}_{s,\rm loc}^{(0)} 
&=\mathrm{Linear}\!\left(\mathbf{X}_s^{\rm loc};\mathbf{W}_{1}^{\rm loc},\mathbf{b}_{1}^{\rm loc}\right)
\in\mathbb{R}^{K\times d_{\rm h}^{\rm loc}},\\
\mathbf{H}_{s,\rm loc}^{(\ell)} 
&\!=\!\operatorname{LN}\!\Big(\!\!\operatorname{ReLU}\!\big(\!\operatorname{MDE}^{\rm loc}_{\ell}(\mathbf{H}_{s,\rm loc}^{(\ell-1)})\big)\!\Big),\ell\!=\!1,\ldots,L,\\
\mathbf{T}_s^{\rm loc} 
&=\mathrm{Linear}\!\left(\mathbf{H}_{s,\rm loc}^{(L)};\mathbf{W}_{2}^{\rm loc},\mathbf{b}_{2}^{\rm loc}\right),
\end{align}\end{subequations}
where \(D^{\rm loc}=8+2M^2+2N^2\), and \(d_{\rm h}^{\rm loc}\) denotes the hidden dimension of the local TEN network. \(\mathbf{W}_1^{\rm loc}\in\mathbb{R}^{D^{\rm loc}\times d_{\rm h}^{\rm loc}}\), \(\mathbf{W}_2^{\rm loc}\in\mathbb{R}^{d_{\rm h}^{\rm loc}\times F^{\rm loc}}\), and \(\mathbf{b}_1^{\rm loc}\in\mathbb{R}^{d_{\rm h}^{\rm loc}},\mathbf{b}_{2}^{\rm loc}\in\mathbb{R}^{F^{\rm loc}}\). 
For the other-satellite input tensors, we apply a 2D-equivariant TEN (with respect to both the other-satellite and UT dimensions):
\begin{align}
\tT_s^{\rm oth}=N_{\rm TEN}^{\rm oth}\big(\tA_s^{\rm oth},\tU_s^{\rm oth}\big) \in \mathbb{R}^{(S-1) \times K \times F^{\rm oth}},
\end{align} 
The network \(N_{\rm TEN}^{\rm oth}\) shares the same building blocks as \(N_{\rm TEN}^{\rm loc}\):
\vspace{-4mm}\begin{subequations}
\begin{align}
\tX_s^{\rm oth}
&=[\tA_s^{\rm oth},\tU_s^{\rm oth}]_3\in\mathbb{R}^{(S-1)\times K\times D^{\rm oth}},\\
\tH_{s,\rm oth}^{(0)} 
&=\!\mathrm{Linear}\!\left(\!\tX_s^{\rm oth};\mathbf{W}_{1}^{\rm oth},\mathbf{b}_{1}^{\rm oth}\!\right)
\!\in\!\mathbb{R}^{(S-1)\times K\times d_{\rm h}^{\rm oth}},\\
\tH_{s,\rm oth}^{(\ell)} 
&=\!\operatorname{LN}\!\Big(\!\!\operatorname{ReLU}\!\big(\!\operatorname{MDE}^{\rm oth}_{\ell}(\tH_{s,\rm oth}^{(\ell-1)})\big)\!\Big),\ell\!=\!1,\ldots,L,\\
\tT_s^{\rm oth} 
&=\mathrm{Linear}\!\left(\tH_{s,\rm oth}^{(L)};\mathbf{W}_{2}^{\rm oth},\mathbf{b}_{2}^{\rm oth}\right).
\end{align}\end{subequations}
\(D^{\rm oth}=6+2M^2+2N^2\), and \(d_{\rm h}^{\rm oth}\) denotes the hidden dimension of the other-satellite TEN network.
\(\mathbf{W}_1^{\rm oth}\in\mathbb{R}^{D^{\rm oth}\times d_{\rm h}^{\rm oth}}\), \(\mathbf{W}_2^{\rm oth}\in\mathbb{R}^{d_{\rm h}^{\rm oth}\times F^{\rm oth}}\), and \(\mathbf{b}_1^{\rm oth}\in\mathbb{R}^{d_{\rm h}^{\rm oth}},\mathbf{b}_{2}^{\rm oth}\in\mathbb{R}^{F^{\rm oth}}\). 
To remove the dependence on the ordering of the other-satellite set, we apply an MDI module  $N_{\rm MDI}^{\rm oth}(\cdot)=\mathrm{MDI}_{\{1\}}(\cdot)$ that performs permutation-invariant attention pooling along the other-satellite dimension:
\begin{align}
\mathbf{T}_s^{\rm oth} = N_{\rm MDI}^{\rm oth}(\tT_s^{\rm oth}) \in \mathbb{R}^{K\times F^{\rm oth}}.
\end{align}
Finally, we fuse the two branches by concatenation along the feature dimension:
\begin{align}
\mathbf{T}_s^{\rm dec}=[\mathbf{T}_s^{\rm loc},\mathbf{T}_s^{\rm oth}]_2
\in\mathbb{R}^{K\times(F^{\rm loc}+F^{\rm oth})}.
\end{align}
The local branch is permutation-equivariant over the UT dimension, while the other branch is permutation-equivariant over both the other-satellite and UT dimensions and becomes permutation-invariant over the other-satellite dimension after MDI aggregation.
\subsubsection{Decentralized FDA Module}
Given the fused features \(\mathbf{T}_s^{\rm dec}\), we apply a decentralized FDA to obtain \(\mathbf{C}_s^{\rm dec}\in\mathbb{R}^{K\times G}\):
\begin{align}
\mathbf{C}_s^{\rm dec}=N_{\rm FDA}^{\rm dec}(\mathbf{T}_s^{\rm dec})\in\mathbb{R}^{K\times G}.
\end{align}
The decentralized FDA module also comprises an LN layer, two linear layers, a GELU activation, and dropout:
\begin{subequations}
\begin{align}
\hat{\mathbf{T}}_{s,\rm FDA}^{\rm dec} 
&= \operatorname{LN}(\mathbf{T}_s^{\rm dec}),\\
\mathbf{H}_{s,\rm FDA}^{\rm dec} 
&= \operatorname{GELU}\!\Big(\mathrm{Linear}\!\left(\hat{\mathbf{T}}_{s,\rm FDA}^{\rm dec};\mathbf{W}_{1}^{\rm dec},\mathbf{b}_{1}^{\rm dec}\right)\Big),\\
\tilde{\mathbf{H}}_{s,\rm FDA}^{\rm dec} 
&= \operatorname{Dropout}\!\left(\mathbf{H}_{s,\rm FDA}^{\rm dec}\right),\\
\mathbf{C}_s^{\rm dec} 
&= \mathrm{Linear}\!\left(\tilde{\mathbf{H}}_{s,\rm FDA}^{\rm dec};\mathbf{W}_{2}^{\rm dec},\mathbf{b}_{2}^{\rm dec}\right),
\end{align}\end{subequations}
 we have \(\mathbf{W}_{1}^{\rm dec}\in\mathbb{R}^{(F^{\rm loc}+F^{\rm oth})\times G}\),
\(\mathbf{W}_{2}^{\rm dec}\in\mathbb{R}^{G\times G}\),
and \(\mathbf{b}_{1}^{\rm dec},\mathbf{b}_{2}^{\rm dec}\in\mathbb{R}^{G}\).
The FDA performs dimension mapping, transforming the fused features into the desired dimension.
\subsubsection{Decentralized CFR Module}
The predicted tuple for the \(s\)-th satellite is obtained via the closed-form recovery mapping:
\begin{align}
\mathsf{O}_s^{\rm dec}=N_{\rm CFR}^{\rm dec}\big(\mathbf{C}_s^{\rm dec}\big).
\end{align}
The recovered variables are then substituted into \eqref{re_precoding} to reconstruct the decentralized precoders \(\{\mathbf{p}_{s,k}\}_{\forall k}\) and receive vectors \(\{\mathbf{b}_{s,k}\}_{\forall k}\). In this way, each satellite performs local inference independently, eliminating the need for centralized CSI collection and precoder dissemination, thereby enabling fully decentralized operation.

Overall, the decentralized network establishes an end-to-end prediction framework that preserves physical consistency and interpretability while enabling coordinated precoding across the multi-satellite massive MIMO system, as summarized in \algref{net-wmmse-decentralized}. Its per-satellite inference complexity consists of four components: the local-satellite TEN branch, with complexity \(\mathcal{C}_{\rm TEN}^{\rm loc}=\mathcal{O}(K(D^{\rm loc} d_{\rm h}^{\rm loc}+L\rho_{1}(d_{\rm h}^{\rm loc})^{2}+d_{\rm h}^{\rm loc}F^{\rm loc}))\); the other-satellite TEN branch, with \(\mathcal{C}_{\rm TEN}^{\rm oth}=\mathcal{O}((S-1)K(D^{\rm oth} d_{\rm h}^{\rm oth}+L\rho_{2}(d_{\rm h}^{\rm oth})^{2}+d_{\rm h}^{\rm oth}F^{\rm oth}))\); the MDI aggregation and FDA head, with \(\mathcal{C}_{\rm MDI}^{\rm oth}=\mathcal{O}(K(S-1)(F^{\rm oth})^{2})\) and \(\mathcal{C}_{\rm FDA}=\mathcal{O}(K((F^{\rm loc}+F^{\rm oth})G+G^{2}))\), respectively; and the CFR module, with \(\mathcal{C}_{\rm CFR}=\mathcal{O}(KM^{3})\), which typically dominates the per-satellite inference cost. Since all satellites perform inference locally, the overall computation can be executed in parallel across satellites, and the practical complexity is therefore determined by the per-satellite cost.

\begin{algorithm}[t]
  \caption{Decentralized TFC-Net-Based WMMSE}
  \label{net-wmmse-decentralized}
  \begin{algorithmic}[1]
    \STATE \textbf{Input:} $\{\{\Phi_{s,k},\Theta_{s,k},\beta_{s,k}, \kappa_{s,k},\boldsymbol{\Sigma}_{s,k},\sigma_{k}^2\}_{\forall k } \cup P_s^{\rm sat}\}_{\forall s}$, \\
    $\{\{\{\Phi_{t,m},\Theta_{t,m},\sigma_{m}^2\}_{\forall m} \cup P_t^{\rm sat}\}_{\forall t \in \mathcal S \setminus s}\}_{\forall s}$.
    \STATE Construct $\{\mathbf{A}_s^{\rm loc},\mathbf{R}_s^{\rm loc},\mathbf{B}_s^{\rm loc},\tA_s^{\rm oth},\tU_s^{\rm oth}\}_{\forall s}$.
    \STATE \textbf{Training:}
    \STATE \!$\mathcal{N}_{\rm dec}^{(0)}\!=\!\!\{\!N_{\rm TEN}^{\rm loc}, N_{\rm TEN}^{\rm oth}, N_{\rm MDI}^{\rm oth}, N_{\rm FDA}^{\rm dec},N_{\rm CFR}^{\rm dec}\!\}^{(0)}\!$ and $n\!=\!0$.
    \STATE \textbf{for} $n < N_{\rm ep}$ \textbf{do}
      \STATE \quad \textbf{forward:}
      \STATE \quad\quad \textbf{for} each satellite $s \in \mathcal S$ \textbf{do}
      \STATE \quad\quad\quad Apply DB-TEN: $\mathbf{T}^{\rm loc}_s = N_{\rm TEN}^{\rm loc}\big(\mathbf{A}^{\rm loc}_s,\mathbf{R}^{\rm loc}_s,\mathbf{B}^{\rm loc}_s\big)$,\\ \quad\quad\quad \quad\quad\quad \quad\quad \quad\quad
      $\tT_s^{\rm oth}  = N_{\rm TEN}^{\rm oth}\big(\tA_s^{\rm oth},\tU_s^{\rm oth}\big)$,\\
      \quad\quad\quad \quad\quad\quad \quad \quad\quad\quad 
      $\mathbf{T}_s^{\rm oth}=N_{\rm MDI}^{\rm oth}\big(\tT_s^{\rm oth}\big)$.
      \STATE \quad\quad\quad Apply FDA mudule: $\mathbf{C}^{\rm dec}_s = N_{\rm FDA}^{\rm dec}\big(\mathbf{T}^{\rm dec}_s\big)$.
      \STATE \quad\quad\quad Recover variables: $\mathsf{O}^{\rm dec}_s = N_{\rm CFR}^{\rm dec}\big(\mathbf{C}^{\rm dec}_s\big)$.
      \STATE \quad\quad\quad Get $\{\mathbf{p}_{s,k}, \mathbf{b}_{s,k}\}_{\forall k}$ from $\mathsf{O}^{\rm dec}_s$ with \eqref{re_precoding}.
      \STATE \quad\quad \textbf{end for}
      \STATE \quad \textbf{backward:}
      \STATE \quad\quad Update $\mathcal{N}_{\rm dec}^{(n)}$ via gradient descent on $\mathcal{L}_{\rm wsr}$. 
     \STATE \quad \quad Update $n=n+1$.
    \STATE \textbf{end for}
   \STATE Store the learned variables of $\mathcal{N}_{\rm dec}$.
    \STATE \textbf{Inference:}  Use $\mathcal{N}_{\rm dec}$ to obtain $\{\mathbf{p}_{s,k},\mathbf{b}_{s,k}\}_{\forall s,k}$.
  \end{algorithmic}
\end{algorithm}
\vspace{-2mm}\subsection{Loss Function for Unsupervised Learning}
In the unsupervised learning paradigm, the network is optimized directly toward the physical objective of maximizing the system sum rate, rather than relying on pre-computed closed-form labels. 
At each training iteration, the model predicts the precoding vectors 
$\{\mathbf{p}_{s,k}\}_{\forall s,k}$ 
and receive vectors 
$\{\mathbf{b}_{s,k}\}_{\forall s,k}$ 
for all satellite-UT pairs. The achievable rate is evaluated based on the iCSI.
The unsupervised objective is formulated as the negative sample-averaged WSR:
\begin{align}
\mathcal{L}_{\rm wsr} 
= -\frac{1}{N_{\rm sam}}\sum_{n=1}^{N_{\rm sam}}\sum_{s=1}^{S}\sum_{k=1}^{K}
a_{s,k}\,R_{s,k}^{(n)},
\end{align}
where $N_{\rm sam}$ represents the sample number. 
This loss function is fully differentiable with respect to all network outputs, allowing end-to-end optimization through standard backpropagation. 
During training, the optimizer minimizes $\mathcal{L}_{\rm wsr}$, which is equivalent to maximizing the achievable system sum rate of the multi-satellite system.
\vspace{-2mm}\subsection{Dataset and Training Details}
The dataset comprises 10,000 samples, with 7,000 used for training, 2,000 for validation, and an additional 1,000 for testing. Each sample randomly selects the center of the cooperative transmission region within the global coverage of the constellation, generates UTs, and determines the cooperating satellites. \textcolor{black}{It further includes data corresponding to five transmit power levels, namely $P_s^{\rm sat} \in  [-10, -5, 0, 5, 10]\, \text{dBW}$, which are randomly selected during training to enable the network to operate under various link budgets resulting from various transceiver configurations \cite{Xiang2024,3GPPTR38.821,wang2025MSMS}.}
\vspace{-2mm}
\section{Numerical Results}\label{simulation}
We employ the QuaDRiGa channel simulator to generate the propagation scenarios and radio channel parameters. Specifically, the channel parameters are obtained using the QuaDRiGa\_NTN\_Urban\_LOS\_scenario \cite{Burkhardt2014}. This simulator, with appropriate parameter calibration, is consistent with the channel model considered in this work as well as the Third Generation Partnership Project specifications. Monte Carlo simulations are conducted by randomly selecting a point within the constellation coverage as the center, defining a circular region with the specified radius, and choosing the \(S\) nearest satellites to that center for multi-satellite cooperative transmission. In addition, all satellites are assumed to have identical power constraints, i.e., $P_s^{\rm sat} = P^{\rm sat}, \forall s$.
The remaining simulation parameters are summarized in \tabref{param}.
	\begin{table}[!t]
    \centering
    \caption{Multi-Satellite System Parameters \cite{FCC_SpaceX_Gen2_2021,3GPPTR38.821,Xiang2024,wang2026}}
    \label{param}
    \resizebox{0.7\columnwidth}{!}{%
    \begin{tabular}{cc}
    \toprule
    \textbf{Parameter}  &  \textbf{Value}  \\
    \midrule
    Constellation Type & Walker-Delta \\
    Satellite Orbit Altitude & $600 \ \rm{km}$ \\
    Orbital Planes & 28 \\
    Satellites Per Plane & 60 \\
    Inclination  & $53^{\circ}$\\
    Multi-Satellite Coverage Radius & $800 \ \rm{km}$ \\
    Number of Cooperating Satellites & 2,3,4\\
    \midrule
    Number of UTs & 6,12,18,24,30 \\ 
    Distribution of UTs & Uniform\\
    Velocity of UTs & $5 \ \rm{km/h}$ \\
    \midrule
    Carrier Frequency & $2\ \rm{GHz}$\\
    Subcarrier Spacing & $30 \ \rm {kHz}$ \\
    System Bandwidth (DL) & $20\ \rm{MHz}$\\
    \midrule
    Tranmit Antenna Number $M$ & $64$\\
    Receive Antenna Number $N$ & $4$\\
    Per-Antenna Gain $ G_{\rm T}$, $G_{\rm R}$ & $6 \ \rm{dBi}$, $ 0\ \rm{dBi}$ \\
    Noise Figure $F$ & $7 \ \rm {dB}$\\
    Noise Temperature $T$ & $290 \ \rm {K}$\\
    \bottomrule
    \end{tabular}%
    }
\end{table}
\subsection{Performance Comparison}
This subsection compares the following schemes:
\begin{itemize}
\item \textbf{Sep-MRT and Sep-MMSE}: Conventional MRT and MMSE are applied independently at each satellite without multi-satellite cooperation, following \cite{Peel2005}. The receive vectors are set to the receive steering vectors.
    \item \textbf{Sep-Opt-WM}: Separate WMMSE carried in each satellite without cooperation among multi-satellites, which ignores inter-satellite interference \cite{Christensen2008, Shi2011}.  
\item \textbf{Cen-Opt-WM}: A centralized WMMSE scheme that computes the precoders through alternating optimization based on sCSI, as in \secref{centralized_optimization}, similar to \cite{Xiang2024}.
\item \textbf{Cen-TFC-WM}: A centralized TFC-Net-based WMMSE scheme, where all satellites upload their local sCSI to a central satellite, which jointly infers the precoding vectors using the centralized network described in \secref{centralized_network}.
\item \textbf{Dec-TFC-WM}: A decentralized TFC-Net-based WMMSE scheme, where each satellite performs local learning and inference, enabling coordinated precoding without central satellite CSI collection or precoder dissemination, as in \secref{decentralized_network}.
\end{itemize}

\figref{1_wsr_pt_3_12} compares the sum-rate performance of different schemes versus the transmit power with $S=3$ and $K=12$. The multi-satellite cooperative transmission schemes that explicitly account for inter-satellite interference, i.e., Cen-Opt-WM, Cen-TFC-WM, and Dec-TFC-WM, significantly outperform the separate-satellite baselines, including Sep-Opt-WM, Sep-MMSE, and Sep-MRT. In particular, the proposed Cen-TFC-WM achieves a sum rate close to the optimal Cen-Opt-WM with substantially lower online inference complexity than iterative optimization. Moreover, Dec-TFC-WM provides substantial gains, and its advantage over Sep-Opt-WM becomes more pronounced as $P^{\rm sat}$ increases. Notably, Dec-TFC-WM can achieve most of the sum-rate performance of Cen-TFC-WM, demonstrating the effectiveness of the proposed decentralized scheme. 
\begin{figure}[!t]
		\centering  \vspace{-4mm}
\includegraphics[width=0.75\linewidth,trim=3.5cm 9.4cm 3.5cm 9.5cm,clip]{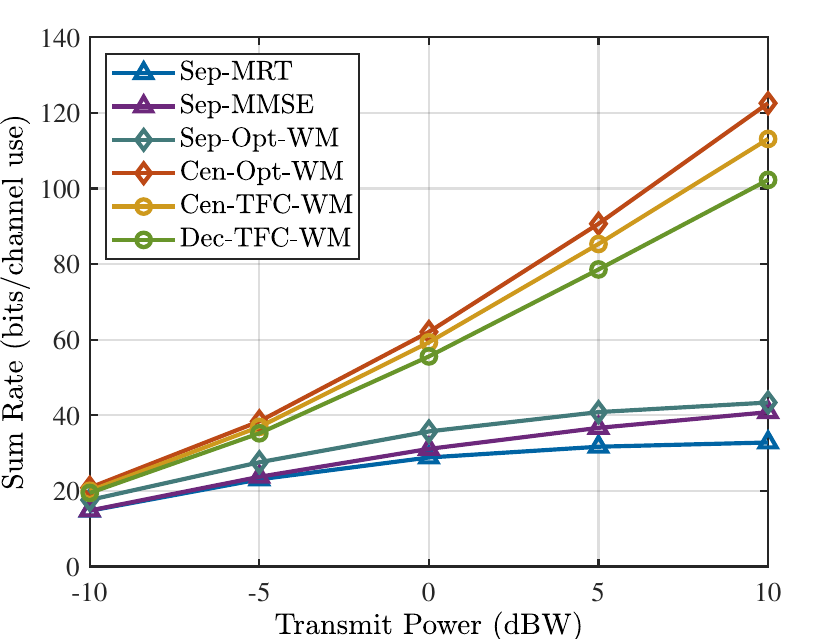}	 \vspace{-2mm}
\caption{Multi-satellite sum rate vs ${P}^{\rm sat}$.}
	\label{1_wsr_pt_3_12}	\vspace{-3mm}
\end{figure}
  \begin{figure}[!t]
		\centering
\includegraphics[width=0.75\linewidth,trim=0.01cm 9.3cm 0.01cm 8.5cm,clip]{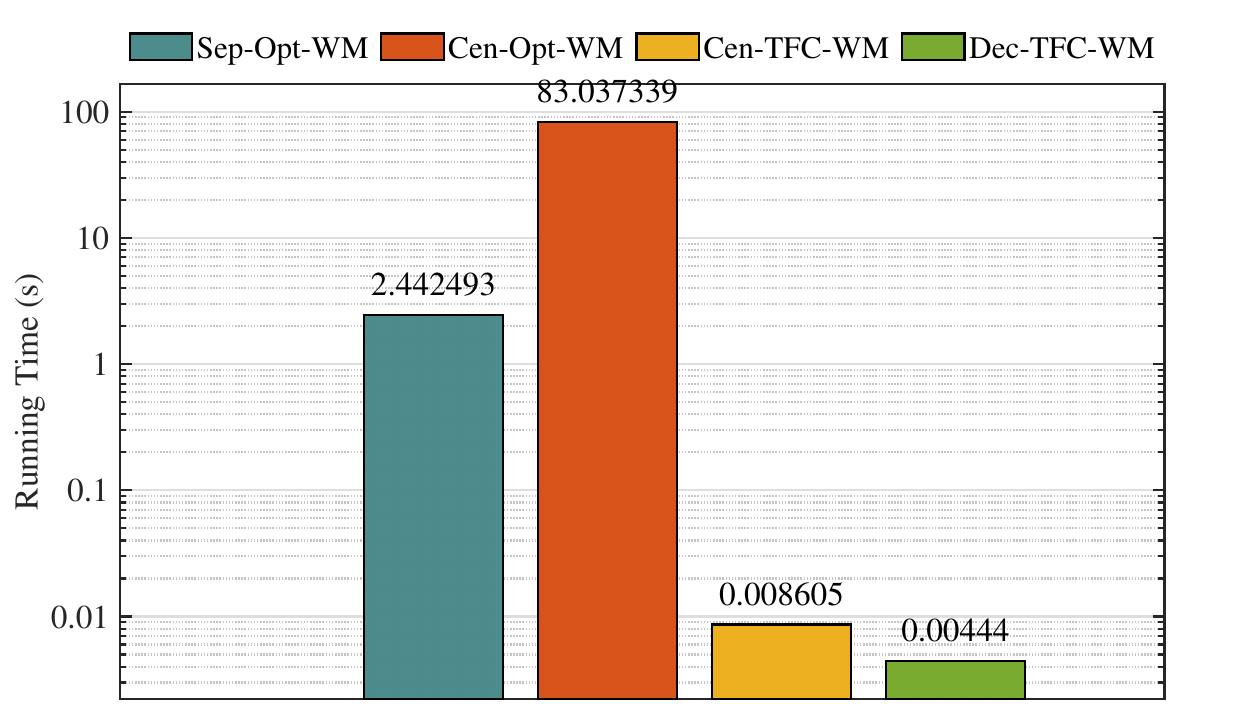}	
 \vspace{-2mm} \caption{Running time comparison.}
\label{1_runtime}	\vspace{-5mm}
	\end{figure} 
    \begin{table}[!t]
\centering
\caption{Computational Complexity Order Comparison.}
\label{complexity}
\resizebox{0.7\columnwidth}{!}{%
\begin{tabular}{lc}
\toprule
\textbf{Schemes} & \textbf{Complexity}\\
\midrule
Sep-Opt-WM 
& $\mathcal{O}\!\left(I_{\rm max}^{\rm out} I_{\rm max}^{\rm in}K M^3\right)$ in parallel \\
Cen-Opt-WM 
& $\mathcal{O}\!\left(I_{\rm max}^{\rm out} I_{\rm max}^{\rm in}SK M^3 \right)$ \\
Cen-TFC-WM
& $\mathcal{O}\!\left(SK M^3\right)$ \\
Dec-TFC-WM 
& $\mathcal{O}\!\left(KM^3\right)$ in parallel \\
\bottomrule
\end{tabular}}
\end{table}    

In \tabref{complexity}, we summarize the computational complexity of the four WMMSE-based schemes. The complexities of Sep-MRT and Sep-MMSE are omitted since they are significantly lower than those of the WMMSE variants. The optimization-based baselines require iterative updates with both outer \(I_{\rm max}^{\rm out}\) and inner \(I_{\rm max}^{\rm in}\) loops, leading to \(\mathcal{O}\left(I_{\rm max}^{\rm out} I_{\rm max}^{\rm in} K M^3\right)\) for Sep-Opt-WM and \(\mathcal{O}\left(I_{\rm max}^{\rm out} I_{\rm max}^{\rm in} S K M^3\right)\) for Cen-Opt-WM, where the cubic term stems from solving \(M\times M\) linear systems. In contrast, the proposed learning-based schemes only require low-complexity inference to avoid iterative optimization, resulting in \(\mathcal{O}(SKM^3)\) for Cen-TFC-WM. For Dec-TFC-WM, each satellite performs inference locally with \(\mathcal{O}(KM^3)\) complexity, and the overall complexity can be further reduced through parallel execution across satellites. \figref{1_runtime} compares the running time of several methods, where the time cost of Sep-MRT and Sep-MMSE is relatively small and thus omitted. The running time experiments are conducted on a CPU platform (Intel(R) Xeon(R) Platinum 8336C CPU @ 2.30GHz). As observed, Cen-TFC-WM and Dec-TFC-WM achieve significantly lower running time due to their non-iterative nature, even outperforming the single-satellite baseline Sep-Opt-WM. Note that running time is influenced by various factors such as low-level implementation efficiency and hardware conditions, and should therefore be seen as indicative rather than absolute.
\begin{figure}[!t]
    \centering \vspace{-6mm}
\includegraphics[width=0.65\linewidth,trim=3.5cm 9.3cm 3.5cm 9.9cm,clip]{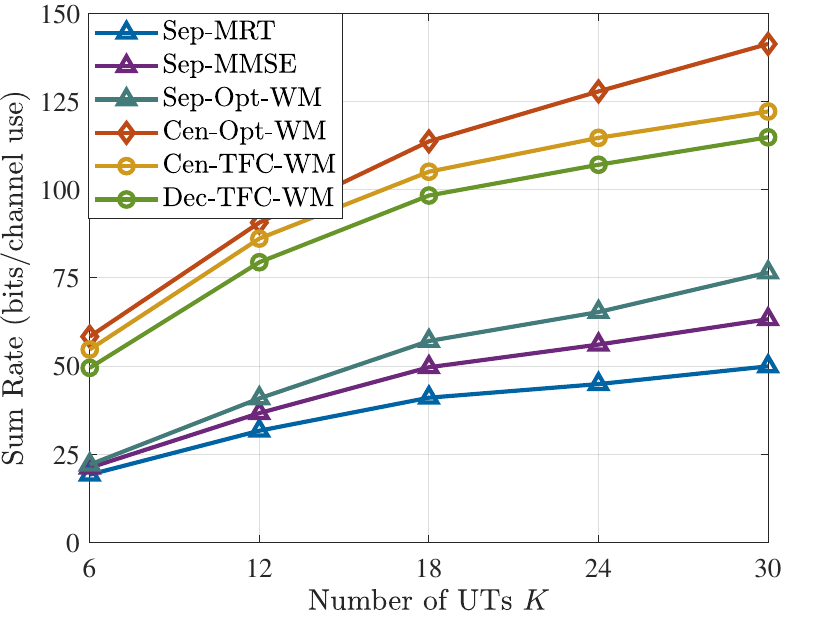} \vspace{-2mm}
    \caption{Multi-satellite sum rate vs UT number with $S=3$ and $P^{\rm sat}=5\, \rm{dBW}$, where Cen-TFC-WM and Dec-TFC-WM are trained at $K=12$.}
\label{fig:3_wsr_sat3_ut_5dbw}\vspace{-3mm}
\end{figure}
\begin{figure}[!t]
    \centering
\includegraphics[width=0.97\linewidth,trim=0.1cm 11cm 0.1cm 12cm,clip]{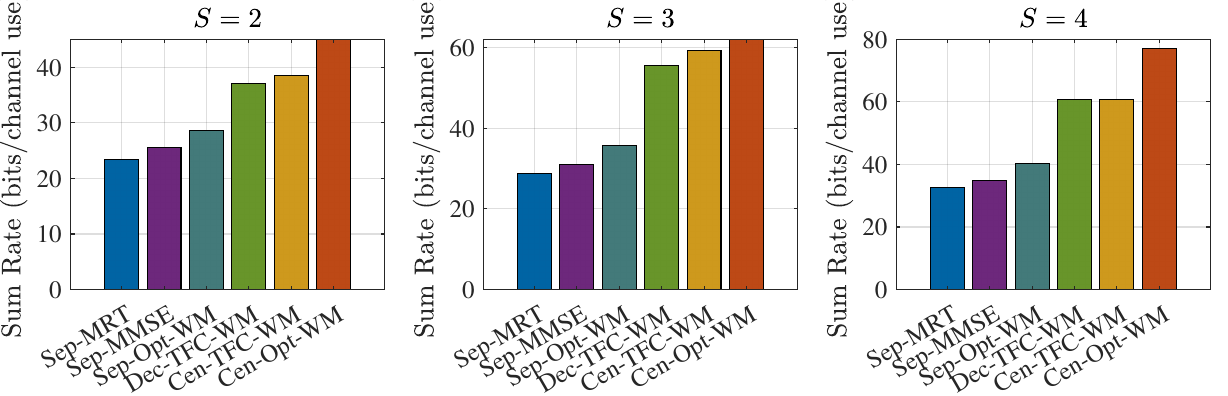}
     \vspace{-3mm}
    \caption{Multi-satellite sum rate versus satellite number with $K\!=\!12$ and $P^{\rm sat}\!=\!5\, \rm{dBW}$, where Cen-TFC-WM and Dec-TFC-WM are trained at $S\!=\!3$.}
\label{fig:2_wsr_sat_ut12_5dbw}\vspace{-3mm}
\end{figure}
\vspace{-8mm}
\subsection{Scenario Scalability}
This subsection evaluates the scenario scalability of the proposed learning-based networks under varying network scales. Specifically, \figref{fig:3_wsr_sat3_ut_5dbw} investigates the impact of increasing the number of UTs with a fixed number of satellites, showing that the proposed schemes maintain stable performance gains as the user load grows, while Cen-TFC-WM exhibits a small degradation relative to Cen-Opt-WM under heavier loads. Moreover, \figref{fig:2_wsr_sat_ut12_5dbw} examines the effect of increasing the number of cooperative satellites with a fixed number of UTs, where Cen-TFC-WM and Dec-TFC-WM consistently improve the sum rate as more satellites participate in cooperation.  Notably, although the networks are trained only under the specific setting \((S=3, K=12)\), they generalize well across a broad range of scenario configurations without retraining.
\vspace{-3mm}
\subsection{Inter-Satellite Overhead}
We summarize the inter-satellite information required by different schemes in \tabref{inter_sat_overhead}.  For Cen-Opt-WM and  Cen-TFC-WM, each satellite needs to report sCSI to the central satellite, which then computes and feeds back the precoding vectors to all cooperating satellites.  The exchanged information is reduced to slowly varying sCSI together with satellite position, attitude, and power constraints.  By contrast, Dec-TFC-WM only requires the exchange of satellite position, attitude, and power information among satellites, since the precoders are generated locally via decentralized inference, eliminating both CSI reporting and centralized precoder dissemination.
\begin{table}[!t]
\centering
\footnotesize % 或 \small / \scriptsize
\caption{Comparison of Inter-Satellite Overhead}
\label{inter_sat_overhead}
\begin{tabular}{p{1.7cm}p{6.3cm}}
\toprule
\textbf{Schemes} & \textbf{Inter-satellite exchanged information} \\
\midrule
Cen-Opt-WM
Cen-TFC-WM
&
\textit{To central satellite}:
satellite position, and attitude;
$\{\beta_{s,k}, \kappa_{s,k}, \boldsymbol{\Sigma}_{s,k}\}_{\forall s \in \mathcal{S}\setminus s, k}$, and
$\{P_s^{\rm sat}\}_{\forall s \in \mathcal{S}\setminus s}$.
\textit{From central satellite}: precoding vectors $\{\mathbf{p}_{s,k}\}_{\forall s, \forall k}$.\\
\midrule
Dec-TFC-WM
&
\textit{Between satellites}: satellite position, attitude, and $P_s^{\rm sat}$. \\
\bottomrule
\end{tabular}
\end{table}
\begin{figure}[!t]
		\centering \vspace{-6mm}
\includegraphics[width=0.72\linewidth,trim=0.5cm 9.1cm 0.5cm 8.7cm,clip] {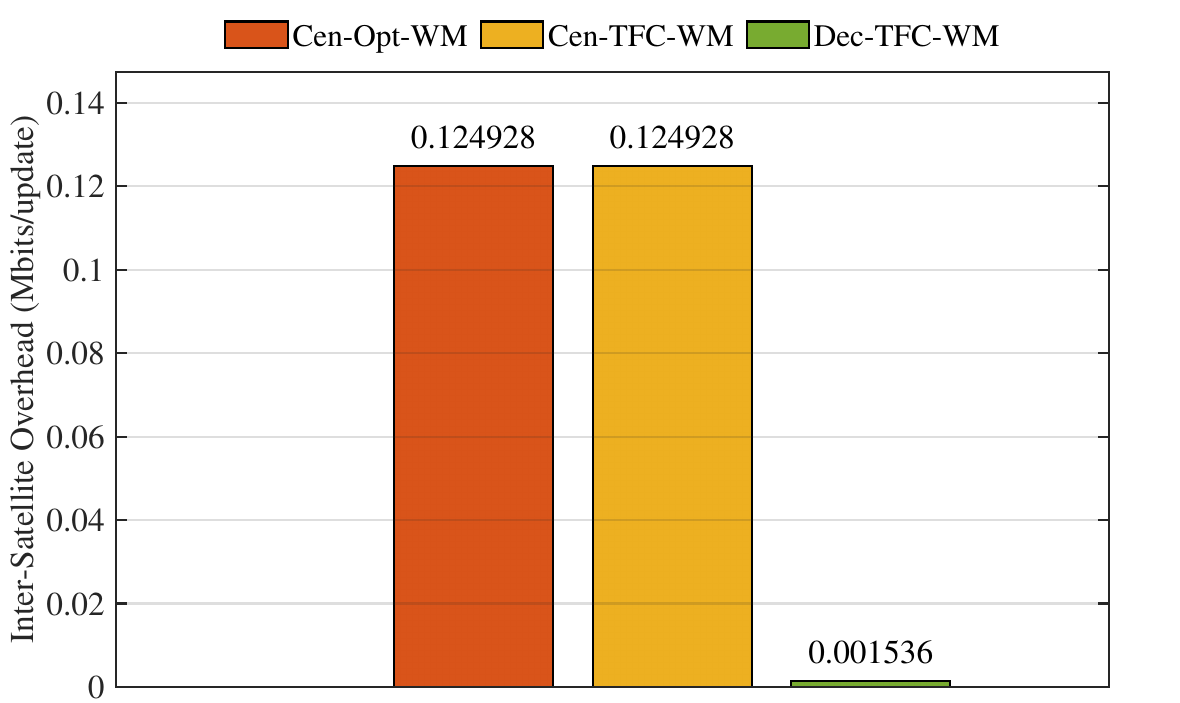}	  \vspace{-2mm}	\caption{Per-update inter-satellite overhead comparison.}
	\label{4_inter_sat_overhead}	\vspace{-5mm}
	\end{figure}
\figref{4_inter_sat_overhead} further quantifies the inter-satellite overhead per update.  Compared with Cen-Opt-WM and Cen-TFC-WM, the proposed Dec-TFC-WM achieves much lower overhead. The reduction stems from removing the need for CSI upload and centralized precoder delivery in Dec-TFC-WM. 
Beyond the per-update cost, the average inter-satellite overhead is determined by both the overhead per update and the update frequency: compared to frequent updates due to iCSI, our proposed schemes can update less often with slowly varying parameters. Specifically, Cen-Opt-WM and Cen-TFC-WM are based on sCSI, and Dec-TFC-WM permits the longest update period since it relies mainly on ephemeris and attitude information. Consequently, the proposed schemes greatly improve inter-satellite communication efficiency and are more practical for multi-satellite massive MIMO transmission.
\vspace{-2mm}\section{Conclusion}
\label{conclusion}
This paper investigated efficient architectures for multi-satellite massive MIMO transmission. \textcolor{black}{We study the WSR maximization in a multi-satellite system where multiple satellites transmit independent data streams to multi-antenna UTs, thereby enhancing system throughput.} By reformulating the WSR maximization problem as a multi-satellite WMMSE problem, closed-form updates for the precoders and receivers were derived. To reduce the complexity of iterative optimization, a centralized learning-based design based on TFC-Net was proposed, achieving strong performance with low-complexity online inference. To further reduce inter-satellite overhead, a decentralized learning-based WMMSE scheme was developed, where each satellite locally infers its precoders using only periodically exchanged SSI and UT GNSS. Numerical results demonstrated that the proposed multi-satellite schemes significantly outperform the single-satellite scheme, while the decentralized scheme achieves near-centralized performance with much lower ISL overhead, and the learning-based schemes exhibit strong robustness and scalability.

	% or
	

	% \section*{Acknowledgment}
	
	%The authors would like to thank Linfeng Song for his valuable suggestions.
	
%\newpage
\appendices

	\ifCLASSOPTIONcaptionsoff
	\newpage
	\fi
	\bibliographystyle{IEEEtran}
	\bibliography{IEEEfull}

\end{document}